\documentclass{article}

\usepackage[nonatbib, final]{neurips_2020}

\usepackage[utf8]{inputenc} 
\usepackage[T1]{fontenc}    
\usepackage{hyperref}       
\usepackage{url}            
\usepackage{booktabs}       
\usepackage{amsfonts}       
\usepackage{nicefrac}       
\usepackage{microtype}      

\usepackage[citestyle=numeric, bibstyle=numeric, uniquename=false, uniquelist=false, maxcitenames=2, maxbibnames=10, giveninits=true]{biblatex}
\usepackage{amsthm}
\usepackage[ruled,vlined]{algorithm2e}
\usepackage{hyperref} 
\usepackage[utf8]{inputenc}
\usepackage[mathscr]{euscript}
\usepackage{graphicx}
\usepackage{caption}
\usepackage{subcaption}
\usepackage{wrapfig}
\usepackage{amsmath}
\usepackage{amsthm}
\usepackage{amssymb}
\usepackage{thmtools}
\usepackage{thm-restate} 
\usepackage{enumitem} 
\usepackage{dsfont} 
\usepackage{tikz-network} 
\usepackage{float} 
\usepackage{xspace}
\usepackage{hyperref} 

\newcommand{\cc}[1]{\mathcal{#1}}
\newcommand{\bb}[1]{\mathbb{#1}}

\newcommand{\indep}{\perp \!\!\! \perp}
\newcommand{\ind}[1]{\mathds{1}\left \{ #1 \right \}}

\newtheorem{definition}{Definition}[section]
\newtheorem{corollary}{Corollary}[section]

\theoremstyle{definition}
\newtheorem{example}{Example}[section]

\newcommand{\aicp}{A-ICP\xspace}

\title{Active Invariant Causal Prediction: Experiment Selection through Stability}

\author{%
  Juan L. Gamella\\
  Seminar for Statistics\\
  ETH Zurich\\
  Switzerland\\
  \texttt{gajuan@ethz.ch} \\
  \And
  Christina Heinze-Deml\\
  Seminar for Statistics\\
  ETH Zurich\\
  Switzerland\\
  \texttt{heinzedeml@stat.math.ethz.ch} \\

}

\begin{document}

\maketitle

\begin{abstract}
    A fundamental difficulty of causal learning is that causal models can generally not be fully identified based on observational data only. Interventional data, that is, data originating from different experimental environments, improves identifiability. However, the improvement depends critically on the target and nature of the interventions carried out in each experiment. Since in real applications experiments tend to be costly, there is a need to perform the \emph{right} interventions such that as few as possible are required. 
    In this work we propose a new active learning (i.e.\ experiment selection) framework (\aicp) based on Invariant Causal Prediction (ICP) \cite{peters2016causal}.
    For general structural causal models, we characterize the effect of interventions on so-called stable sets, a notion introduced by \cite{pfister2019stabilizing}. We leverage these results to propose several intervention selection policies for \aicp which quickly reveal the direct causes of a response variable in the causal graph while maintaining the error control inherent in ICP. Empirically, we analyze the performance of the proposed policies in both population and finite-regime experiments. 
\end{abstract}

\section{Introduction}

Causal models \cite{pearl2009causal} capture the causal relationships between variables and allow us to predict how a system behaves under interventions or distribution changes. Hence, they are more powerful than probabilistic models, and can be seen as abstractions of more accurate mechanistic or physical models while retaining enough power to answer interventional or counterfactual questions \cite{peters2017elements}. Therefore, they maintain their predictive power in new, previously unseen environments \cite{HeinzeDeml2017, Subbaswamy2018LearningPM, Rothenhausler2018AnchorRH, pfister2019stabilizing}. 

The question remains if for systems of interest such models can be learned directly from data. This problem is known in the literature as \emph{causal learning}, and it is to causal models what statistical learning is to probabilistic models. Just like statistical learning, it suffers from the inherent difficulty of determining properties of a distribution from finite-sized samples. Additionally, causal learning is challenged by the fact that, even with full knowledge of the underlying observational distribution, some causal relationships cannot be established and causal models can generally not be fully identified from observational data alone \cite{pearl2009causal}. For causal directed acyclic graph models, this limit of identifiability implies that, from observational data alone, the true graph cannot be distinguished from others that lie in the same Markov equivalence class \cite{verma1990equivalence}. Under additional assumptions about the model class and noise distributions, full identifiability is still possible \cite{hoyer2009nonlinear, buhlmann2014cam, shimizu2006linear, peters2013equalvar, peters2011identifiability}. In the general case, however, identifiability can only be improved by performing interventions (experiments). Examples of such interventions are abundant in the empirical sciences, from gene knockout experiments in biology to chemical compound selection in drug discovery \cite{naik2016active}. Since such experiments tend to be costly, there is a need to pick the \emph{right} interventions, in the sense of having to do as few of them as possible. In the remainder of this section, we review existing work that addresses this problem before summarizing our contributions.



\subsection{Related work}

We use the term \emph{active causal learning} to refer to learning causal models from data while being able to actively perform interventions. In this setting, the goal is to sequentially improve identifiability, as opposed to the classical setting in machine learning \cite{settles2009active}, where the goal is to sequentially increase prediction accuracy. Existing approaches can be said to fall broadly into two categories: Bayesian and graph-theoretic. The Bayesian approach, pioneered by the works of \cite{tong2001active, murphy2001active}, selects interventions which maximize a Bayesian utility function, generally the mutual information between the graph and the hypothetical sample that the experiment would produce. More recent works build on this approach by considering experiments performed in batches under budget constraints \cite{agrawal2019abcd} or when expert knowledge is available \cite{masegosa2013interactive}, and apply such framework to learning biological networks \cite{cho2016reconstructing, ness2017bayesian}.

Among the graph-theoretic approaches, \cite{eberhardt2008almost, hyttinen13a} give bounds on the number of interventions required for identifiability under different assumptions, and \cite{he2008active, hauser2014two} provide intervention selection strategies that aim to orient the maximum number of edges in the graph. There are extensions to several settings, such as when the total number of interventions is limited \cite{ghassami2018budgeted, ghassami2019interventional}, when there are hidden variables \cite{NIPS2017_7277} or when interventions carry a cost which must be minimized \cite{kocaoglu2017cost, lindgren2018experimental}.

Both approaches make different assumptions and suffer from different drawbacks. The Bayesian approach requires exact knowledge of the intervention location and parameters. It is difficult to analyze the impact of misspecified interventions on the choice of experiments and the estimates produced by the methods \cite{walker2013bayesian}. Furthermore, it suffers from poor computational scaling \cite{ness2017bayesian} and several approximations have to be made even for small graphs \cite{agrawal2019abcd}. This further complicates giving guarantees on the result. Graph-theoretic approaches are agnostic to the underlying distribution, but generally make two strong assumptions: (1) that the Markov equivalence class has been correctly identified, which is difficult with a limited sample size, and (2) that interventions are perfectly informative (i.e.\ infinite interventional data).

\textbf{Invariant Causal Prediction and intervention stable sets} 
This work is a first attempt at a new approach which falls into neither of the previous two categories. It relies on Invariant Causal Prediction (ICP) \cite{peters2016causal}, which aims to recover the direct causes $S^*$ of a \emph{response} variable of interest $Y$  from interventional data. The general idea is that the conditional distribution of the response, given its direct causes, remains invariant when intervening on arbitrary variables in the system other than itself. ICP considers the setting where different experimental conditions of a system exist (called \emph{environments}) and an i.i.d.\ sample of each environment is available. By considering all possible subsets of the predictor variables $X$, ICP then searches for sets of \emph{plausible causal predictors}. These are sets of predictors which, if conditioned on, leave the distribution of the response invariant across the observed environments (see \autoref{sec:causal_predictors} for the formal definition). This procedure is based on testing the null hypothesis of invariance. Sets considered as plausible causal predictors given the available data are referred to as \emph{accepted sets}, and the set of direct causes of the response (its parents in the causal graph) will be among them with high probability. ICP then returns the intersection of all accepted sets as the estimate $\hat{S}$ of the direct causes. More details are given in \autoref{apx:alg}.

While ICP does not retrieve the full graph, it has some important advantages in the form of guarantees and more flexible assumptions. It requires neither knowledge of the Markov equivalence class, nor about the nature or location of the interventions performed in each environment, except that they must not act on the response. The approach assumes that the noise distribution of the response is independent from the direct causes and invariant across environments. In the general formulation, no further distributional assumptions are made. While such further assumptions can arise from the choice of tests for the invariance of the conditional distribution, non-parametric tests can be chosen \cite{heinze2018invariant}. Perhaps most importantly, ICP provides an error control with respect to the estimated causes, namely that with high probability it will not retrieve false positives. While this comes at a loss of power, this work shows that when environments are generated via an appropriate experiment selection strategy, ICP can quickly identify the direct causes while maintaining the aforementioned control. In addition to ICP, we make use of the notion of so-called \emph{intervention stable sets} \cite{pfister2019stabilizing} which relates the invariance properties of a set of predictors to graphical criteria. More details are given in \autoref{sec:stable_sets}.

\subsection{Contributions and outline}
We propose Active Invariant Causal Prediction (\aicp), an active causal learning framework based on ICP. \autoref{fig_procedure} shows its core components. In each round, a new intervention target is selected based on the sets accepted by ICP in the previous iteration. Subsequently, the corresponding experiment is performed, yielding a new sample of interventional data\footnote{While \autoref{fig_procedure} is illustrated using do-intervention notation, noise or shift interventions are also possible, on a single or multiple variables.}. 
Finally, ICP is run on the updated dataset which yields updated estimates of the accepted sets and the direct causes of $Y$. 
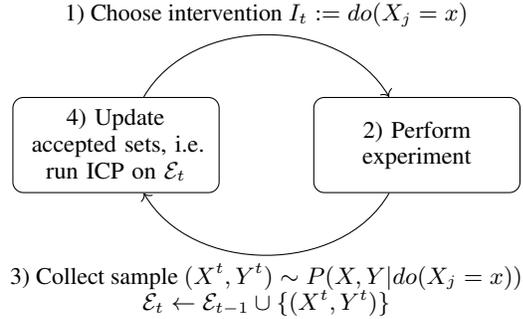
\begin{wrapfigure}{r}{0.51\textwidth}
\vspace{-2mm}
\small
    \centering
    \tikzstyle{block} = [rectangle, draw, text width=8em, text centered, rounded corners, minimum height=4em, minimum width=6em]
    \begin{tikzpicture}[node distance = 2cm, auto]
        \node [block] (init) {4) Update accepted sets, i.e.\ run ICP on $\cc{E}_t$};
        \node [block, right of=init, node distance=4cm] (update) {2) Perform experiment};
        \draw [->] (init) to [out=60,in=120] node {1) Choose intervention $I_t := do(X_j = x)$} (update);
        \draw [->] (update) to [out=-120,in=-60] node[align=center] {3) Collect sample $(X^t, Y^t) \sim P(X,Y|do(X_j=x))$\\$\cc{E}_t \leftarrow \cc{E}_{t-1} \cup \{(X^t, Y^t)\}$ } (init);
    \end{tikzpicture}
    \caption{Schematic of \aicp}\label{fig_procedure}
\end{wrapfigure} 
Our main contribution lies in the formulation of several policies that choose an intervention target in each round $t$. 
They are motivated by theoretical results on the invariance properties of sets of predictors (\autoref{sec:stable_sets}). In \autoref{sec:causal_predictors}, we detail how we combine these results with ICP in an active causal learning setting. 
We then propose several policies that fit into the \aicp framework in \autoref{sec:policies}. While our theoretical results do not require any parametric assumptions on the underlying structural causal model (SCM), we focus on linear SCMs in the empirical evaluation in \autoref{sec:experiments}. In population and finite regime experiments, the proposed policies outperform a random baseline policy across a large range of experimental settings. Finally, we compare \aicp against ABCD \cite{agrawal2019abcd} and discuss the observed tradeoffs between error control and power.

\section{Intervention stable sets}\label{sec:stable_sets}
We now present the theoretical results that motivate the intervention selection policies in each round $t$ of \aicp.  
We use the framework of structural causal models (SCMs) \cite{Wright1921, Haavelmo1944, Bollen1989}. A SCM consists of (i) a collection of structural assignments that functionally relate each variable in the system to its direct causes and (ii) a joint distribution over the noise variables which are required to be jointly independent. A SCM induces a joint distribution over the variables in the system as well as a graph over the associated vertices (e.g.\ see Definition 6.2 in \cite{peters2017elements}).
In the following setting, we formalize the assumptions required for the results derived in this section. Importantly, for the results presented in this section, we do not require the SCM to be linear. All proofs can be found in \autoref{apx:proofs}.

\paragraph{Setting 1} (adapted from setting 2 in \cite{pfister2019stabilizing}) Let $X \in \boldsymbol{\cc{X}} = \cc{X}_1 \times ... \times \cc{X}_p$ be predictor variables, $Y \in \mathbb{R}$ a \emph{response} variable and $I = (I_1,...,I_m) \in \boldsymbol{\cc{I}} = \cc{I}_1 \times ... \times \cc{I}_m$ intervention variables which are unobserved and formalize the interventions present in the collection of intervention environments $\cc{E}$. Assume there exists a SCM $\cc{S}^\cc{E}$ over $(I,X,Y)$ that can be represented by a directed acyclic graph $\cc{G}(\cc{S}^\cc{E})$, in which the intervention variables are source nodes. Further assume intervention variables do not appear in the structural equation of $Y$, that is, assume there are no interventions on the response. For each $e\in\cc{E}$, there is a SCM $\cc{S}_e$ over $(I^e, X^e, Y^e)$ such that $\cc{G}(\cc{S}_e)=\cc{G}(\cc{S}^\cc{E})$, in which only the equations with $I^e$ on the right hand side change with respect to $\cc{S}^\cc{E}.$ Furthermore, assume that the distribution of $(I^e, X^e, Y^e)$ is absolutely continuous with respect to a product measure that factorizes.

No further assumptions are made on the size or type of the intervention, i.e.\ they can be do, noise or shift interventions on a single or multiple variables. To simplify notation, let $\text{PA}(i)$, $\text{CH}(i)$ and $\text{AN}(i)$ be the parents, children and ancestors of the variable $X_i$, respectively.

The notion of intervention stable sets, introduced in \cite{pfister2019stabilizing}, allows characterizing sets of plausible causal predictors from d-separation relationships in the graph. While stable sets are generally not equivalent to the sets of plausible causal predictors, we here derive theoretical results for them and then analyze under which conditions these apply to the plausible causal predictors (\autoref{sec:causal_predictors}).

\begin{definition}[intervention stable set \cite{pfister2019stabilizing}] Let for any
set $S \subseteq\{1,...,p\}$, $X_S$ be the vector containing all variables $X_k, k \in S$. Given setting 1 and a set of environments $\cc{E}$, we call a set $S \subseteq \{1,...,p\}$ \emph{intervention stable} under $\cc{E}$ if the d-separation $I \indep_\cc{G} Y \mid X_S$ holds in $\cc{G}(\cc{S}^\cc{E})$ for any intervention $I$ which is active in an environment $e \in \cc{E}$.
\end{definition}

In other words, a set of predictors is stable if it d-separates the response from all interventions (see \autoref{example_stable_sets}). In the following, let $\bb{S}_\cc{E}$ denote the collection of sets which are intervention stable under $\cc{E}$.
The stable sets allow properties of the graph structure and the interventions to be inferred:

\begin{restatable}[intervened parents appear on all intervention stable sets]{lemma}{lemmaparents}
\label{lem:parents_stable}
Let $\cc{E}$ be a set of observed environments and let $j \in \text{PA}(Y)$ be directly intervened on in $\cc{E}$. Then,
$$S \subseteq \{1,...,p\} \text{ is intervention stable} \implies j \in S.$$
\end{restatable}

\begin{restatable}[sets containing descendants of directly intervened children are unstable]{lemma}{lemmachildren}
\label{lem:children_unstable}
Let $i \in \text{CH}(Y)$ be directly intervened on in $\cc{E}$. Then, any set $S \subseteq \{1,...,p\}$ which contains descendants of $i$ is not intervention stable.
\end{restatable}

\begin{restatable}[stability of the empty set]{lemma}{lemmaempty}
\label{lem:empty_set}
Let $\cc{E}$ be any set of environments. Then,
$$\emptyset \in \bb{S}_\cc{E} \iff \cc{E} \text{ contains no interventions on variables in } \text{AN}(Y).$$
\end{restatable}
That is, the empty set is stable if and only if none of the interventions in $\cc{E}$ occurred upstream of $Y$. More structure can be inferred by considering the number of stable sets in which a predictor appears:

\begin{definition}[stability ratio]
\label{def:stab_ratio}
Given a set of environments $\cc{E}$, the \emph{stability ratio} of a variable $i \in \{1,...,p\}$ is defined as
\begin{align*}
    r_\cc{E}(i) := \frac{1}{|\bb{S}_\cc{E}|}\sum_{S\in\bb{S}_\cc{E}} \ind{i \in S},
\end{align*}
i.e.\ the proportion it appears in the intervention stable sets under $\cc{E}$.
\end{definition}

From \autoref{lem:parents_stable} it follows that parents which are directly intervened on in at least one environment in $\cc{E}$ have a stability ratio of 1. Conversely, by \autoref{lem:children_unstable} descendants of children directly intervened on in at least one environment have a ratio of 0. Furthermore, the stability ratio of any ancestor, regardless of the interventions, is always larger than one half:

\begin{restatable}[ancestors appear on at least half of all stable sets]{proposition}{propratio}
\label{prop:ratio_ancestors}
Let $\cc{E}$ be any set of observed environments. Then, for any $j \in \{1,...,p\}$,
$$r_{\cc{E}}(j) < 1/2 \implies j \notin \text{AN}(Y).$$
\end{restatable}
\begin{corollary}
\label{cor:ratio_parents}
The parents of the response always have a stability ratio of or above $1/2$.
\end{corollary}

Note that the converse is not generally true, i.e.\ variables which are not ancestors can have a stability ratio of or above $1/2$, even after being intervened on. In \autoref{sec:policies} we exploit \autoref{lem:parents_stable}, \autoref{lem:empty_set} and \autoref{cor:ratio_parents} to construct intervention selection policies.

\section{From stable sets to causal predictors}
\label{sec:causal_predictors}
The results derived in \autoref{sec:stable_sets} apply to intervention stable sets. If we are to use these results to construct an intervention selection policy for \aicp, we need to understand under which conditions they apply directly to the sets of plausible causal predictors. 

\begin{definition}[plausible causal predictors \cite{peters2016causal}]\label{def:plausible_causal_predictors}  We call a set of variables $S \subseteq\{1,...,p\}$ \emph{plausible causal predictors} under a set of environments $\cc{E}$ if for all $e, f \in \cc{E}$ and all $x$
\begin{align}\label{icp_hypothesis}
    Y^e \vert {X^e_S=x} \;\; \stackrel{d}{=} \;\; Y^f \vert {X^f_S=x}, 
\end{align}
i.e.\ the conditional distribution is the same in all environments. Let $\bb{C}_ \cc{E}$ denote the collection of sets which are plausible causal predictors under $\cc{E}$.
\end{definition}

Given a collection of environments $\cc{E}$, the collection of accepted sets of the ICP algorithm is an estimate of $\bb{C}_\cc{E}$.  The following proposition establishes the relationship between intervention stable sets and sets of plausible causal predictors.

\begin{restatable}[intervention stable sets are plausible causal predictors]{proposition}{propcausalpredictors}
\label{prop:causal_predictors}
Let $\cc{E}$ be a set of observed environments. Then, for all intervention stable sets $S \subseteq \{1,...,p\}$, it holds that $S\in\bb{C}_\cc{E}$.
\end{restatable}

While $\bb{S}_\cc{E} \subseteq \bb{C}_\cc{E}$, it is not generally true that $\bb{S}_\cc{E} = \bb{C}_\cc{E}$, even under the faithfulness assumption (see \autoref{example_stable_vs_plausible}). However, when the parameters of the SCM are sampled from a continuous distribution, we conjecture that the set of parameters for which $\bb{S}_\cc{E} \neq \bb{C}_\cc{E}$ has probability zero. We call the assumption that  $\bb{S}_\cc{E} = \bb{C}_\cc{E}$ \emph{stability-faithfulness}, and adopt this assumption in the following.

Finally, we make use of the following corollary in \aicp. In each iteration an intervention target is selected and a sample is collected from the new experimental environment (see \autoref{fig_procedure}). Denote by $\cc{E}_t = \left \{e_i : i \in \{1,..,t\} \right \}$ the set of observed environments at iteration $t$, and assume $\cc{E}_t \subseteq \cc{E}_{t+1}$. 

\begin{corollary}\label{cor_speedup}
Let $\cc{E}_t, \cc{E}_{t+1}$ be sets of observed environments such that $\cc{E}_t \subseteq \cc{E}_{t+1}$. Then, it follows that if $S$ is not a set of plausible causal predictors under $\cc{E}_t$, it is not under $\cc{E}_{t+1}$ either.
\end{corollary}

\section{Constructing an active learning policy}\label{sec:policies}
Even in the population setting---in the absence of estimation errors---the capacity of ICP to retrieve the parents relies heavily on the informativeness of the environments. For example, if none of the interventions are upstream of the response, the empty set is intervention stable and is returned as estimate of the parents.
While \cite{peters2016causal} gives some sufficient conditions for the identifiability of the true causal parents, it is not entirely clear what an optimal intervention is. If we assume stability-faithfulness, by \autoref{lem:parents_stable} we know that, in the absence of estimation errors, a direct intervention on a parent is sufficient for it to appear in the ICP estimate. However, it is not a necessary condition (see \autoref{example_dir_intervention}). As a first approach, we treat direct interventions on the parents as ``maximally informative'', and the goal of the proposed policies is to pick such interventions. For simplicity and to allow comparison with ABCD \cite{agrawal2019abcd}, we consider single-variable interventions.


\paragraph{Proposed policies} 
To increase the chances of picking a parent of the response as an intervention target, the proposed policies can make use of three strategies:
\begin{enumerate}[itemsep=.5ex, leftmargin=.5cm]
    \item (\emph{Markov strategy}, (``markov'')) This strategy selects intervention targets from within the Markov blanket, which contains the parents. Under linearity, in the population setting the Markov blanket can be directly obtained from an ordinary least squares (OLS) regression over all predictors (\autoref{apx:experiments}). In the finite regime, we turn to the Lasso \cite{tibshirani1996regression} to obtain an estimate.
    \item (\emph{empty-set strategy}, (``e'')) If an observational sample is available, we can test whether the invariance in Eq.~\eqref{icp_hypothesis} holds for the empty set when considering the observational and the interventional sample $e_t$. If it does, by \autoref{lem:empty_set} we know that the latest intervention target is not upstream of the response, and therefore not a parent. We hence discard the target from future interventions.
    \item (\emph{ratio strategy}, (``r'')) By \autoref{cor:ratio_parents}, a variable is not a parent if it appears on less than half of all intervention stable sets. As an estimate we use the accepted sets (computed based on the environments $\cc{E}_{t-1}$) and, if a variable appears on less than half of such sets, we do not add it to the pool of possible intervention targets for the current iteration. Note that unlike in (2.), we do not discard it from future interventions. This is important in the finite regime, where parents may for some iterations appear in less than half of all accepted sets due to testing errors.
\end{enumerate}
Furthermore, we exclude identified parents, i.e.\ variables with a stability ratio of 1, from the pool of possible intervention targets for all of the above strategies.
Each strategy narrows down the set of possible intervention targets, and the actual target is then chosen uniformly at random.
For multiple-variable interventions we can simply pick $k$ targets instead of one. If a policy combines several of the above strategies, the final set of possible intervention targets is taken as the intersection of the respective strategies' sets. 
For instance, by combining the \emph{ratio} and the \emph{empty-set} strategies we increase the chance of picking a parent but also exclude non-ancestors which retain a stability ratio above one half after being intervened on. As non-ancestors can have a stability ratio above one half, strategies exploiting the value of the stability ratio (e.g.\ by intervening on the variable whose value of $r_\cc{E}(i)$ is closest to $1/2$ in absolute value) are not competitive.
\newpage
\paragraph{Algorithm outline} In the \aicp framework, policies are treated as interchangeable modules which define two functions: \texttt{next\_intervention} and \texttt{first\_intervention}. At each iteration of the procedure, the accepted sets (given the current environments $\cc{E}_t$) are passed to the policy by calling \texttt{next\_intervention}, which then returns the next intervention target.
Since the accepted sets are not available when selecting the initial intervention, a potentially available (observational) sample can be used to guide a first choice. For example, policies employing the \emph{Markov} strategy compute an estimate of the Markov blanket in this step, and pick a variable within this estimate as the first intervention; other policies pick an intervention at random. For simplicity, we assume the availability of an initial observational sample, even though this is not necessary for all of the proposed strategies.

By \autoref{cor_speedup}, at each iteration of \aicp it suffices for ICP to consider only the sets accepted in the previous iteration, which provides a substantial speed up as not all subsets of predictors need to be re-tested. To account for multiple testing of the accepted sets, we need to apply a correction to the significance level of ICP. Due to the strong dependence between the tests, we use a Bonferroni correction and run ICP at a significance level of $\alpha/T$, where $\alpha$ is the desired overall level and $T$ is the total number of iterations for which \aicp is run. Hence, the coverage guarantee for the final \aicp estimate $\hat{S}(\cc{E}_T)$ is $P(\hat{S}(\cc{E}_T) \subseteq S^*) \geq 1-\alpha$. 
Further details and an outline of ICP can be found in \autoref{apx:alg}. More details about the multiple testing correction are given in \autoref{apx:50_iter_fwer} and a sensitivity analysis with respect to the chosen overall significance level $\alpha$ can be found in Appendix~\ref{apx:figures:level_strength}. 
Finally,  an analysis of the computational complexity of algorithm \ref{alg_active_icp} is given in \autoref{apx:complexity}.

\begin{algorithm}[H]
\label{alg_active_icp}
\SetAlgoLined
\SetKwInOut{Input}{Input}
\SetKwInOut{Output}{Output}
\Output{$\hat{S}(\cc{E}_T)$ estimate of the parents of the response}
\Input{\texttt{policy} an intervention selection policy,\newline
$(X^0, Y^0)$ sample from initial environment,\newline
$T$ number of iterations,\newline
$\alpha$ overall \aicp significance level}
$\cc{E}_0 \leftarrow \{(X^0, Y^0)\}$\;
\texttt{accepted sets} $\leftarrow$ all sets of predictors\;
\texttt{next\_intervention} $\leftarrow$ \texttt{policy.first\_intervention}($\cc{E}_0$)\;
\For{$t=1:T$}{
  perform $\texttt{next\_intervention}$ and collect sample $(X^t, Y^t)$\;
  $\cc{E}_t \leftarrow \cc{E}_{t-1} \cup \{(X^t, Y^t)\}$\;
  \texttt{accepted sets}, $\hat{S}(\cc{E}_t)$ $\leftarrow$ \texttt{ICP}($\cc{E}_t$, \texttt{accepted sets}, $\alpha/T$) \tcp*{see \autoref{cor_speedup}}
  \texttt{next\_intervention} $\leftarrow$ \texttt{policy.next\_intervention(accepted sets)}\;
 }
 \Return $\hat{S}(\cc{E}_T)$
 \caption{\aicp}
\end{algorithm}

\section{Experiments}\label{sec:experiments}

We evaluate policies that use different combinations of the strategies in both the population and finite sample setting, using simulated data from randomly chosen linear SCMs. In addition to averaging over different SCMs, for every SCM, each policy is run a number of times with different random seeds to account for the stochastic component of the policies. Further details about the experimental settings and links to code to reproduce the experiments can be found in \autoref{apx:experiments}. \autoref{tab:exp_settings} summarizes the considered settings and experimental parameters.

\subsection{Population setting}\label{subsec:pop_setting}
For the population setting, we evaluate the \emph{Markov} and \emph{ratio} strategies. Population experiments simplify an initial evaluation of the proposed policies. First, since interventions are perfectly informative, the performance of the policies can be compared exclusively in terms of their choice of targets, without worrying about (1) the parameters of the intervention, and (2) how many observations must be allocated to the experiment, neither of which are trivial problems. Second, we can ignore estimation errors, e.g.\ the Markov blanket here simply corresponds to the variables with non-zero coefficient in a population OLS regression over all predictors. Lastly, by \autoref{lem:parents_stable} we have that in the population setting intervening on each predictor variable once is sufficient to produce the correct estimate. This yields a limit on the number of iterations for which \aicp has to be run. Hence, in the population setting we sample without replacement from the pool of possible intervention targets. This is also why the \emph{empty-set} strategy is not applicable in the population setting---we never intervene on the same variable twice in any case. This stands in contrast to the finite regime.

We compare the performance of the two proposed policies (``markov`` and ``markov + r'') with each other and a baseline \emph{random} policy which picks intervention targets at random from all predictors. In \autoref{fig_pop}, we compare how quickly the policies recover the causal parents in terms of the Jaccard similarity between the estimate $\hat{S}$ and the truth $S^*$: $\vert \hat{S} \cap S^* \vert / \vert \hat{S} \cup S^* \vert $. This metric is equal to one if and only if $\hat{S} = S^*$. Furthermore, we assess how many interventions the policies require in total to achieve \emph{exact} recovery. Note that when the Markov blanket is composed of just the parents, both proposed policies are equivalent. This is the case for 405 out of the 1000 SCMs. Therefore, we plot the results for the remaining graphs separately in \autoref{fig_pop} (panels (b) and (d)). In both of the considered metrics, combining the \emph{Markov} and \emph{ratio} strategies produces the best performing policy.

\begin{figure}[H]
\centering
\includegraphics[width=1.0\textwidth]{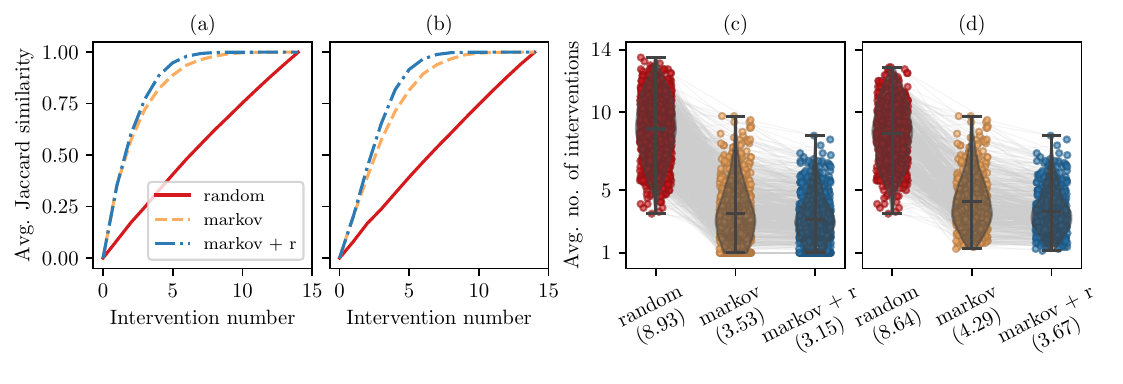}
\caption{(population setting) Left: Average Jaccard similarity as a function of the number of interventions, for all 1000 SCMs of size 15 (a), and those for which the Markov blanket contains more than just the parents (b). In the first intervention the policies ``markov'' and ``markov + r'' perform equally, as they decide among the same pool of intervention targets (the Markov blanket). After interventions on all predictors all direct causes are determined. Right: Average number of interventions needed to achieve exact recovery of the causal parents, for each one of the 1000 SCMs (c) and those for which the Markov blanket contains more than just the parents (d). Each SCM is represented by a dot and connected across policies by a line. The total average of interventions employed by each policy is given below its label.}
\label{fig_pop}
\end{figure}

\subsection{Finite sample setting}

For finite samples, we individually evaluate the effect of the three strategies put forward in \autoref{sec:policies}, as well as combinations of them. In total we have 7 policies, each using a different combination of strategies, plus the random baseline policy. For the sample allocation, we fix the size of the sample collected per intervention; we perform experiments for 10, 100 and 1000 observations per sample. The same metrics as in \autoref{subsec:pop_setting} are shown in Figures~\ref{fig_20_iter_jaccard_finite} and~\ref{fig_20_iter_int_numbers_finite}\footnote{Additionally, we plot the family-wise error rate $\hat{P}(\hat{S} \not\subseteq S^*)$ in \autoref{fig_fwer} (Appendix~\ref{apx:50_iter_fwer}), confirming that the ICP error control is indeed maintained.}. As real experiments tend to be extremely costly in terms of time and money, the goal is to achieve good performance after as few interventions as possible. Hence, we here focus on the first 20 iterations of \aicp while the total number of iterations is 50. Figures~\ref{fig_50_iter_jaccard_finite} and~\ref{fig_50_iter_int_numbers_finite} show the performance over all 50 iterations.

The results show interesting patterns. In general, we observe that the choice of the optimal strategy depends on the number of observations one can allocate to each intervention and how many interventions can be performed in total.
Relying on the Markov blanket estimate (policies labeled with ``markov'') leads to good initial but poor performance for larger $t$, independently of what other strategies are used. This estimate is obtained by performing an L1-regularized least squares regression on all predictors (i.e.\ the Lasso \cite{tibshirani1996regression}), picking the regularization parameter for each SCM by cross validation. In the first few iterations, the policies quickly identify the parents contained in the estimate, as can be seen in \autoref{fig_20_iter_jaccard_finite}. However, when not all parents are contained in the estimate the policies become stuck performing non-informative interventions. As a result, for many of the SCMs not all direct causes are recovered after reaching the limit of iterations, clearly seen for 50 iterations in \autoref{fig_50_iter_int_numbers_finite} (\autoref{apx:figures}). This problem is alleviated at larger sample sizes, where the Lasso yields a better estimate of the Markov blanket. We provide further error analyses in Appendix~\ref{apx:figures:mb_estimation}.

\begin{figure}[t]
\centering
\includegraphics[width=1.0\textwidth]{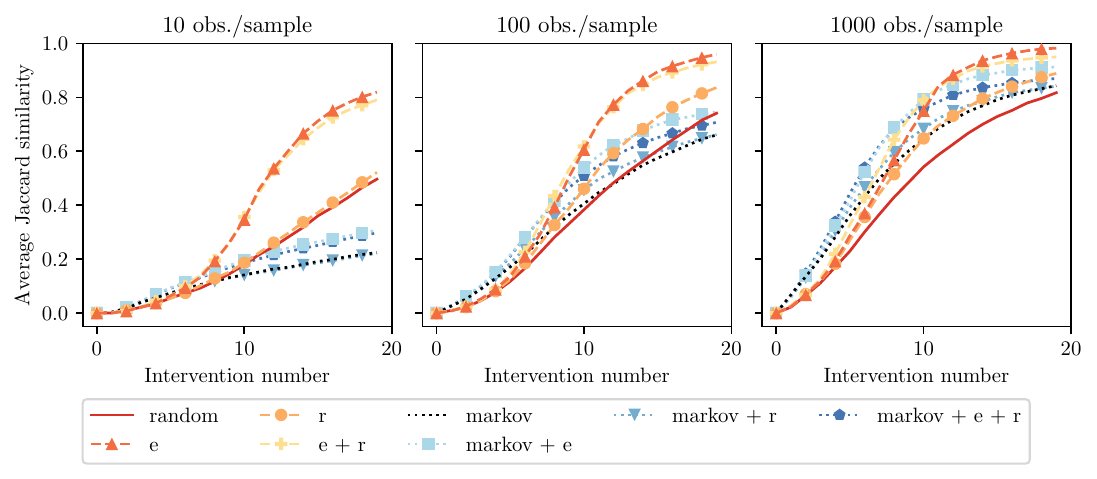}
\caption{(finite regime) Average Jaccard similarity for 300 SCMs of size 12 as a function of the number of interventions for 10, 100 and 1000 observations per sample. Here, $\alpha = 0.01$.}
\label{fig_20_iter_jaccard_finite}
\end{figure}

In principle, the other two strategies may also suffer from estimation errors: For the \emph{empty-set} strategy, the empty set may be wrongly accepted after an intervention on a parent, which is then discarded from future interventions. This problem of statistical power is attenuated for larger sample sizes and higher intervention strengths. For the \emph{ratio} strategy, falsely rejecting stable sets and wrongly accepting unstable sets can bias the estimate of the stability ratio of some parents and thereby keep \aicp from intervening on such parents. 
Hence, if there are no constraints on the number of total interventions, the \emph{random} baseline policy is the most robust option as it is not as affected by estimation errors as the other policies. However, for small $t$, it is clearly outperformed by most of the other policies. As can be seen in \autoref{fig_20_iter_jaccard_finite}, the gain over the \emph{random} policy becomes larger as the sample size increases.

For our experimental settings, we find that the performance of the \emph{empty-set} strategy is quite robust, outperforming the remaining policies across the different sample sizes and a large range of intervention numbers. Using the \emph{Markov} and/or the \emph{ratio} strategy in addition only yields clear improvements for larger sample sizes. In Appendix~\ref{apx:figures:level_strength} we further analyze the performance for different intervention strengths and significance levels. Importantly, while the discussed possible estimation errors affect the choice of the optimal intervention target selection, the ICP error control on the estimate $\hat{S}$ is unaffected by this and remains intact (also see \autoref{fig_fwer}). 

\subsection{Comparison with ABCD}

We compare the performance of \aicp against that of the Bayesian ABCD strategy \cite{agrawal2019abcd}. We choose this strategy as it allows directly learning the parents of the response. It hence lends itself to a more fair comparison than strategies which estimate  the full graph. That said, the comparison is still not straightforward, as both strategies make different assumptions. ABCD requires a large observational sample, which it then uses to sample from the posterior through a bootstrap procedure based on GIES \cite{hauser2012characterization}, but the interventional sample size can be as small as one. \aicp does not rely on a large observational sample but is regression-based and requires more than one observation per intervention. We establish a middle ground by providing both methods with an observational sample of size 1000 and 10 observations per intervention. 
In \autoref{fig_abcd_vs_aicp}, we compare the methods over 50 iterations in terms of (1) the Jaccard similarity, and (2) the family-wise error rate (FWER) $\hat{P}(\hat{S} \not\subseteq S^*)$, i.e.\ the probability of having one or more false positives in the estimate of the causal parents. The results for varying observational sample sizes are similar and can be found in Appendix~\ref{apx:figures:abcd}. Details about the experimental setup are given in \autoref{apx:experiments}.

\begin{figure}[H]
\centering
\includegraphics[width=1.0\textwidth]{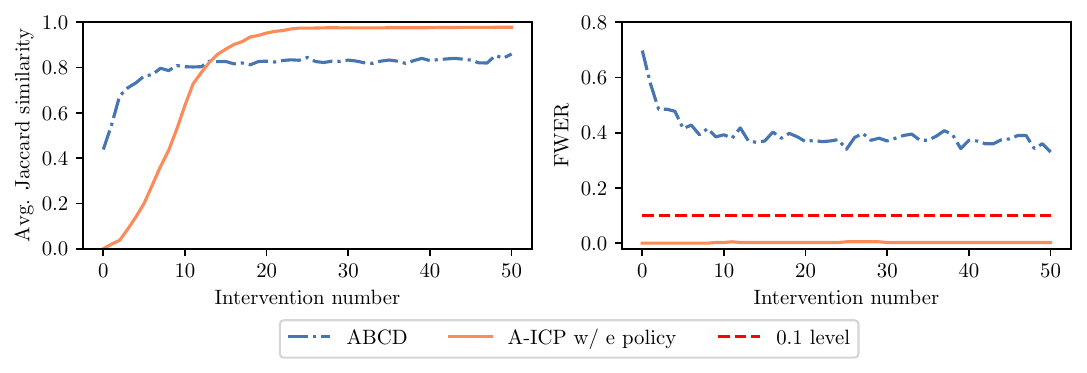}
\caption{Average Jaccard similarity and family-wise error rate of the ABCD and \aicp estimates, as a function of the number of interventions. Exploiting high probability candidates from the posterior over graphs based on the observational sample, ABCD shows better results in terms of Jaccard similarity in the first few iterations while not approaching a Jaccard similarity of one for large $t$. The good initial performance comes at the cost of false positives which are only reduced somewhat as $t$ increases. In contrast, \aicp remains conservative in the first few iterations, often returning the empty set as an estimate. While \aicp retains its error control over all iterations, its power increases steadily with the number of interventions, approaching a Jaccard similarity of one. Here, $\alpha=0.1$.}
\label{fig_abcd_vs_aicp}
\end{figure}

\section{Discussion}
In \autoref{sec:stable_sets} we characterize the effect of interventions on the stability of sets of predictors. We leverage these results to construct several intervention selection policies for \aicp. We find that the  \emph{empty-set} strategy shows good performance across different sample sizes and number of performed interventions. The \emph{ratio} and the \emph{Markov} strategy additionally yield improvements for larger sample sizes and for small intervention numbers. All policies outperform the \emph{random} baseline policy across a large range of settings. While ICP is often criticized for its lack of power, we see that \aicp can quickly overcome this weakness while maintaining the guarantees of ICP.

We welcome the discussion whether the proposed policies for \aicp could be improved as many interesting questions remain. ICP does not require knowledge of the intervention locations in each environment. This makes it robust to interventions with off-target effects, i.e.\ effects on variables other than the target. Furthermore, \aicp  allows for combining data from existing environments with possibly unknown intervention targets with data from experiments that are performed with this knowledge.
On the other hand, one might ask whether, since we know the intervention location when running \aicp, we discard useful information. Of the proposed policies, only the ones that use the \emph{empty-set} strategy leverage this information.

Finally, the results from \autoref{sec:stable_sets} are quite general in the sense that they make no assumptions on the function class or noise distributions of the SCM. As such, it would be interesting to assess to what extent \aicp improves the power of more general extensions of Invariant Causal Prediction, such as nonlinear ICP \cite{heinze2018invariant} or ICP for sequential data \cite{pfister2019invariant}. Overall, this work shows that in an active learning setting, one can construct competitive methods with invariance as the underlying principle for causal discovery.

\newpage
\section*{Discussion of broader impact}
Any method that learns from finite data is subject to statistical estimation errors and model assumptions that necessarily limit the full applicability of its findings. Unfortunately, study outcomes are not always communicated with the required qualifications. As an example, statistical hypothesis testing is often employed carelessly, e.g.\ by using p-values to claim ``statistical significance'' without paying attention to the underlying assumptions \cite{colquhoun2014investigation}. There is a danger that this problem gets exacerbated when one aims to estimate causal structures. Estimates from causal inference algorithms could be claimed to ``prove'' a given causal relationship, ruling out various alternative explanations that one would consider when explaining a statistical association. For example, ethnicity could be claimed to have a causal effect on criminality and thereby used as a justification for oppressive political measures. While this would represent a clear abuse of the technology, we as researchers have to ensure that similar mistakes in interpretation are not made unintentionally. This implies being conscientious about understanding as well as stating the limitations of our research.

While there is a risk that causal inference methods are misused as described above, there is of course also an array of settings where causal learning---and in particular active causal learning---can be extremely useful. As our main motivation we envision the empirical sciences where interventions correspond to physical experiments which can be extremely costly in terms of time and/or money. For complex systems, as for example gene regulatory networks in biology, it might be difficult for human scientists to choose informative experiments, particularly if they are forced to rely on data alone. Our goal is to develop methods to aid scientists to better understand their data and perform more effective experiments, resulting in significant resource savings.
The specific impact of our proposed methodology will depend on the application. For the method we propose in this work, one requirement for application would be that the experiments yield more than one data point (and ideally many), so that our invariance-based approach can be employed. In future work, we aim to develop methodology that is geared towards the setting where only very few data points per experiment are available.

\begin{ack}
We would like to thank Niklas Pfister, Jonas Peters, Armeen Taeb, Brian McWilliams and Nicolai Meinshausen for valuable discussions and comments on the manuscript. The research leading to these results was supported by a grant from the ''la Caixa'' Foundation (ID 100010434), with code LCF/BQ/EU18/11650051.
\end{ack}

\printbibliography

@article{HeinzeDeml2017,
  title={Conditional variance penalties and domain shift robustness},
  author={Christina Heinze-Deml and Nicolai Meinshausen},
  year={2017},
  journal={arXiv preprint arXiv:1710.11469}
}

@article{Rothenhausler2018AnchorRH,
  title={Anchor regression: heterogeneous data meets causality},
  author={Dominik Rothenh\"ausler and Nicolai Meinshausen and Peter B\"uhlmann and Jonas Peters},
  year={2018},
  journal={arXiv preprint arXiv:1801.06229}
}

@article{Subbaswamy2018LearningPM,
  title={Preventing Failures Due to Dataset Shift: Learning Predictive Models That Transport},
  author={Adarsh Subbaswamy and Peter Schulam and Suchi Saria},
  journal={The 22nd International Conference on Artificial Intelligence and Statistics},
  year={2019}
}

@Book{Bollen1989,
  author = 	 {K. A. Bollen},
  title = 	 {Structural Equations with Latent Variables},
  publisher = 	 {John Wiley \& Sons},
  address =      {New York, USA},
  year = 	 {1989}
}

@article{Wright1921,
  author = {S.~Wright},
  journal = {Journal of Agricultural Research},
  pages = {557--585},
  title = {Correlation and Causation},
  volume = 20,
  year = 1921
}

@article{Haavelmo1944,
  author = {T. Haavelmo},
  title = {The Probability Approach in Econometrics},
  journal = {Econometrica},
  volume = {12},
  pages = {S1--S115 (supplement)},
  year = {1944}
}

@article{hauser2014two,
  title={Two optimal strategies for active learning of causal models from interventional data},
  author={Hauser, Alain and B{\"u}hlmann, Peter},
  journal={International Journal of Approximate Reasoning},
  volume={55},
  number={4},
  pages={926--939},
  year={2014},
  publisher={Elsevier}
}

@article{hauser2012characterization,
  title={Characterization and greedy learning of interventional Markov equivalence classes of directed acyclic graphs},
  author={Hauser, Alain and B{\"u}hlmann, Peter},
  journal={Journal of Machine Learning Research},
  volume={13},
  number={Aug},
  pages={2409--2464},
  year={2012}
}

@inproceedings{ness2017bayesian,
  title={A Bayesian active learning experimental design for inferring signaling networks},
  author={Ness, Robert Osazuwa and Sachs, Karen and Mallick, Parag and Vitek, Olga},
  booktitle={International Conference on Research in Computational Molecular Biology},
  pages={134--156},
  year={2017},
  organization={Springer}
}

@inproceedings{ghassami2018budgeted,
  title={Budgeted Experiment Design for Causal Structure Learning},
  author={Ghassami, AmirEmad and Salehkaleybar, Saber and Kiyavash, Negar and Bareinboim, Elias},
  booktitle={International Conference on Machine Learning},
  pages={1719--1728},
  year={2018}
}

@article{cho2016reconstructing,
  title={Reconstructing causal biological networks through active learning},
  author={Cho, Hyunghoon and Berger, Bonnie and Peng, Jian},
  journal={PloS one},
  volume={11},
  number={3},
  pages={e0150611},
  year={2016},
  publisher={Public Library of Science}
}

@inproceedings{tong2001active,
  title={Active learning for structure in Bayesian networks},
  author={Tong, Simon and Koller, Daphne},
  booktitle={International joint conference on artificial intelligence},
  volume={17},
  number={1},
  pages={863--869},
  year={2001},
  organization={Citeseer}
}

@article{murphy2001active,
  title={Active learning of causal Bayes net structure},
  author={Murphy, Kevin P},
  year={2001},
  publisher={Citeseer}
}

@misc{pfister2019stabilizing,
    title={Stabilizing Variable Selection and Regression},
    author={Niklas Pfister and Evan G. Williams and Jonas Peters and Ruedi Aebersold and Peter Bühlmann},
    year={2019},
    eprint={1911.01850},
    archivePrefix={arXiv},
    primaryClass={stat.ME}
}

@article{peters2016causal,
  title={Causal inference by using invariant prediction: identification and confidence intervals},
  author={Peters, Jonas and B{\"u}hlmann, Peter and Meinshausen, Nicolai},
  journal={Journal of the Royal Statistical Society: Series B (Statistical Methodology)},
  volume={78},
  number={5},
  pages={947--1012},
  year={2016},
  publisher={Wiley Online Library}
}

@book{peters2017elements,
  title={Elements of causal inference: foundations and learning algorithms},
  author={Peters, Jonas and Janzing, Dominik and Sch{\"o}lkopf, Bernhard},
  year={2017},
  publisher={MIT press}
}

@article{pearl2009causal,
  title={Causal inference in statistics: An overview},
  author={Pearl, Judea and others},
  journal={Statistics surveys},
  volume={3},
  pages={96--146},
  year={2009},
  publisher={The author, under a Creative Commons Attribution License}
}

@inproceedings{eberhardt2008almost,
  title={Almost optimal intervention sets for causal discovery},
  author={Eberhardt, Frederick},
  booktitle={Proceedings of the Twenty-Fourth Conference on Uncertainty in Artificial Intelligence},
  pages={161--168},
  year={2008}
}

@inproceedings{agrawal2019abcd,
  title={ABCD-Strategy: Budgeted Experimental Design for Targeted Causal Structure Discovery},
  author={Agrawal, Raj and Squires, Chandler and Yang, Karren and Shanmugam, Karthikeyan and Uhler, Caroline},
  booktitle={The 22nd International Conference on Artificial Intelligence and Statistics},
  pages={3400--3409},
  year={2019}
}

@inproceedings{peters2011identifiability,
  title={Identifiability of causal graphs using functional models},
  author={Peters, J and Mooij, J and Janzing, D and Sch{\"o}lkopf, B},
  booktitle={27th Conference on Uncertainty in Artificial Intelligence (UAI 2011)},
  pages={589--598},
  year={2011},
  organization={AUAI Press}
}

@inproceedings{verma1990equivalence,
  title={Equivalence and synthesis of causal models},
  author={Verma, Thomas and Pearl, Judea},
  booktitle={Proceedings of the Sixth Annual Conference on Uncertainty in Artificial Intelligence},
  pages={255--270},
  year={1990},
  organization={Elsevier Science Inc.}
}

@article{he2008active,
  title={Active learning of causal networks with intervention experiments and optimal designs},
  author={He, Yang-Bo and Geng, Zhi},
  journal={Journal of Machine Learning Research},
  volume={9},
  OPTnumber={Nov},
  pages={2523--2547},
  year={2008}
}

@article{masegosa2013interactive,
  title={An interactive approach for Bayesian network learning using domain/expert knowledge},
  author={Masegosa, Andr{\'e}s R and Moral, Seraf{\'\i}n},
  journal={International Journal of Approximate Reasoning},
  volume={54},
  number={8},
  pages={1168--1181},
  year={2013},
  publisher={Elsevier}
}

@techreport{settles2009active,
  title={Active learning literature survey},
  author={Settles, Burr},
  year={2009},
  institution={University of Wisconsin-Madison Department of Computer Sciences}
}

@article{walker2013bayesian,
  title={Bayesian inference with misspecified models},
  author={Walker, Stephen G},
  journal={Journal of Statistical Planning and Inference},
  volume={143},
  number={10},
  pages={1621--1633},
  year={2013},
  publisher={Elsevier}
}

@article{tibshirani1996regression,
  title={Regression shrinkage and selection via the lasso},
  author={Tibshirani, Robert},
  journal={Journal of the Royal Statistical Society: Series B (Methodological)},
  volume={58},
  number={1},
  pages={267--288},
  year={1996},
  publisher={Wiley Online Library}
}

@article{heinze2018invariant,
  title={Invariant causal prediction for nonlinear models},
  author={Heinze-Deml, Christina and Peters, Jonas and Meinshausen, Nicolai},
  journal={Journal of Causal Inference},
  volume={6},
  number={2},
  year={2018},
  publisher={De Gruyter}
}

@article{pfister2019invariant,
  title={Invariant causal prediction for sequential data},
  author={Pfister, Niklas and B{\"u}hlmann, Peter and Peters, Jonas},
  journal={Journal of the American Statistical Association},
  volume={114},
  number={527},
  pages={1264--1276},
  year={2019},
  publisher={Taylor \& Francis}
}

@article{hyttinen13a,
  author  = {Antti Hyttinen and Frederick Eberhardt and Patrik O. Hoyer},
  title   = {Experiment Selection for Causal Discovery},
  journal = {Journal of Machine Learning Research},
  year    = {2013},
  volume  = {14},
  pages   = {3041-3071},
  OPTurl     = {http://jmlr.org/papers/v14/hyttinen13a.html}
}

@article{shimizu2006linear,
  title={A linear non-Gaussian acyclic model for causal discovery},
  author={Shimizu, Shohei and Hoyer, Patrik O and Hyv{\"a}rinen, Aapo and Kerminen, Antti},
  journal={Journal of Machine Learning Research},
  volume={7},
  number={Oct},
  pages={2003--2030},
  year={2006}
}

@article{buhlmann2014cam,
  title={CAM: Causal additive models, high-dimensional order search and penalized regression},
  author={B{\"u}hlmann, Peter and Peters, Jonas and Ernest, Jan and others},
  journal={The Annals of Statistics},
  volume={42},
  number={6},
  pages={2526--2556},
  year={2014},
  publisher={Institute of Mathematical Statistics}
}

@inproceedings{hoyer2009nonlinear,
  title={Nonlinear causal discovery with additive noise models},
  author={Hoyer, Patrik O and Janzing, Dominik and Mooij, Joris M and Peters, Jonas and Sch{\"o}lkopf, Bernhard},
  booktitle={Advances in neural information processing systems},
  pages={689--696},
  year={2009}
}

@article{peters2013equalvar,
   title={Identifiability of Gaussian structural equation models with equal error variances},
   volume={101},
   number={1},
   journal={Biometrika},
   publisher={Oxford University Press (OUP)},
   author={Peters, J. and B{\"u}hlmann, P.},
   year={2013},
   month={11},
   pages={219–228}
}

@article{ghassami2019interventional,
  title={Interventional Experiment Design for Causal Structure Learning},
  author={Ghassami, AmirEmad and Salehkaleybar, Saber and Kiyavash, Negar},
  journal={arXiv preprint arXiv:1910.05651},
  year={2019}
}

@incollection{NIPS2017_7277,
title = {Experimental Design for Learning Causal Graphs with Latent Variables},
author = {Kocaoglu, Murat and Shanmugam, Karthikeyan and Bareinboim, Elias},
booktitle = {Advances in Neural Information Processing Systems 30},
OPTeditor = {I. Guyon and U. V. Luxburg and S. Bengio and H. Wallach and R. Fergus and S. Vishwanathan and R. Garnett},
pages = {7018--7028},
year = {2017},
OPTpublisher = {Curran Associates, Inc.},
OPTurl = {http://papers.nips.cc/paper/7277-experimental-design-for-learning-causal-graphs-with-latent-variables.pdf}
}

@inproceedings{kocaoglu2017cost,
  title={Cost-optimal learning of causal graphs},
  author={Kocaoglu, Murat and Dimakis, Alex and Vishwanath, Sriram},
  booktitle={Proceedings of the 34th International Conference on Machine Learning-Volume 70},
  pages={1875--1884},
  year={2017},
  OPTorganization={JMLR. org}
}

@inproceedings{lindgren2018experimental,
  title={Experimental design for cost-aware learning of causal graphs},
  author={Lindgren, Erik and Kocaoglu, Murat and Dimakis, Alexandros G and Vishwanath, Sriram},
  booktitle={Advances in Neural Information Processing Systems},
  pages={5279--5289},
  year={2018}
}

@article{colquhoun2014investigation,
  title={An investigation of the false discovery rate and the misinterpretation of p-values},
  author={Colquhoun, David},
  journal={Royal Society open science},
  volume={1},
  number={3},
  pages={140216},
  year={2014},
  publisher={The Royal Society Publishing}
}

@article{naik2016active,
  title={Active machine learning-driven experimentation to determine compound effects on protein patterns},
  author={Naik, Armaghan W and Kangas, Joshua D and Sullivan, Devin P and Murphy, Robert F},
  journal={Elife},
  volume={5},
  pages={e10047},
  year={2016},
  publisher={eLife Sciences Publications Limited}
}
\newpage

\appendix
\renewcommand{\thetable}{\thesection.\arabic{table}}
\renewcommand{\thefigure}{\thesection.\arabic{figure}}
\section{Intervention stable sets, plausible causal predictors and informative interventions}
\label{apx:examples}

\subsection{Intervention stable sets}\label{apx:stable_set_example}
A set of predictors $S$ is an intervention stable set if it d-separates the response from all interventions, i.e. if the d-separation statement $I \indep_\cc{G} Y | X_S$ holds in $\cc{G}(\cc{S}^\cc{E})$ for all interventions $I$ active in $\cc{E}$. An example follows:

\begin{center}
\begin{minipage}{.50\textwidth}
\begin{example}
\label{example_stable_sets}

Let $\cc{E}$ be a collection of environments with direct interventions on $X_0$ and $X_4$, as shown in the graph. Then, the intervention stable sets are
\begin{align*}
\bb{S}_\cc{E} =&
    \{0\},\\
    &\{0, 4\},\\
    &\{0, 3, 4\}\\
    &\{0, 1\}\\
    &\{0, 1, 4\}\\
    &\{0, 1, 3, 4\}.
\end{align*}
\end{example}
\end{minipage}
\begin{minipage}{.49\textwidth}
\begin{figure}[H]
    \centering
    \begin{tikzpicture}
        \SetVertexStyle[FillColor=white, MinSize=0.8\DefaultUnit, TextFont=\small]
        \SetEdgeStyle[Arrow=-stealth, TextFont=\small, LineWidth=1]
        \Vertex[y=1, x=-1,label=$X_0$]{0}
        \Vertex[y=1, x=1, label=$X_1$]{1}
        \Vertex[y=0, label=$Y$]{Y}
        \Vertex[y=-1.5,label=$X_3$]{3}
        \Vertex[y=-0.5, x=1, label=$X_4$]{4}
        \Vertex[y=2, x=0, label=$I_1$, shape=diamond, style=white, fontcolor=red]{I1}
        \Vertex[y=0.5, x=2, label=$I_2$, shape=diamond, style=white, fontcolor=red]{I2}
    
        \Edge[Direct](0)(Y)
        \Edge[Direct](1)(Y)
        \Edge[Direct](Y)(3)
        \Edge[Direct](4)(3)
        \Edge[Direct, color=red](I1)(0)
        \Edge[Direct, color=red](I2)(4)
    \end{tikzpicture}
\end{figure}
\end{minipage}
\end{center}

\subsection{Stable sets vs.\ plausible causal predictors}\label{apx:stable_sets_plausible_causal}
While $\bb{S}_\cc{E} \subseteq \bb{C}_\cc{E}$, it is not generally true that $\bb{S}_\cc{E} = \bb{C}_\cc{E}$. Importantly, this does not change when assuming faithfulness as the following example illustrates. 

\begin{example}
\label{example_stable_vs_plausible}
Take the following SCM,

\begin{center}
\begin{minipage}{0.4\textwidth}
\begin{figure}[H]
    \centering
    \begin{tikzpicture}
        \SetVertexStyle[FillColor=white, MinSize=0.8\DefaultUnit, TextFont=\small]
        \SetEdgeStyle[Arrow=-stealth, TextFont=\small, LineWidth=1]
        \Vertex[label=$X_0$]{0}
        \Vertex[x=-2, y=-2, label=$X_1$]{1}
        \Vertex[x=2, y=-2, label=$X_2$]{2}

        \Edge[Direct, label=$w_{01}$](0)(1)
        \Edge[Direct, label=$w_{02}$](0)(2)
        \Edge[Direct, label=$w_{12}$](1)(2)
    \end{tikzpicture}
\end{figure}
\end{minipage}
\begin{minipage}{0.4\textwidth}
\begin{align*}
    X_0 &:= \varepsilon_0\\
    X_1 &:= w_{01}X_0 + \varepsilon_1\\
    X_2 &:= w_{02}X_0 + w_{12}X_1 + \varepsilon_2
\end{align*}
\end{minipage}
\end{center}

with $\varepsilon_i \sim_{\text{i.i.d.}} \mathcal{N}(\mu_i,\sigma^2_i)$ noise variables such that $\varepsilon_i \indep \varepsilon_j\; \forall i,j$. Consider $Y:=X_1$ and the conditioning sets $S_0 = \{0\}$ and $S_2 = \{2\}$. In the following, we assess the invariance of the conditional distributions $Y\vert X_0$ and $Y\vert X_2$ under interventions.
 The conditional distributions of $Y\vert X_0$ and $Y\vert X_2$ are both Gaussian and below we compute their expectations and variances. For $Y\vert X_0$ we have: 
\begin{align*}
    \bb{E}(Y\vert X_0) &=  \bb{E}(Y) + \frac{ \mathrm{Cov}(Y, X_0)}{ \mathrm{Var}(X_0)}  (X_0 - \bb{E}(X_0))\\
     &= w_{01}\mu_0+\mu_1 + \frac{w_{01}\sigma^2_0}{\sigma^2_0}(X_0 - \mu_0) = \mu_1 + w_{01} X_0
\end{align*}
\begin{align*}
    \mathrm{Var}(Y\vert X_0) &=  \mathrm{Var}(Y) - \frac{ \mathrm{Cov}(Y, X_0)^2}{ \mathrm{Var}(X_0)} = w_{01}^2\sigma^2_0 + \sigma^2_1 - \frac{(w_{01}\sigma^2_0)^2}{\sigma^2_0} = \sigma^2_1
\end{align*}

For $Y\vert X_2$ we have: 
\begin{align*}
    \bb{E}(Y\vert X_2) &=  \bb{E}(Y) + \frac{ \mathrm{Cov}(Y, X_2)}{ \mathrm{Var}(X_2)}  (X_2 - \bb{E}(X_2))\\
    & = w_{01}\mu_0+\mu_1 + \frac{\sigma_0^2(w_{01}w_{02}+w_{01}^2w_{12})+w_{12}\sigma^2_1}{\sigma_0^2(w_{02}^2+w_{12}^2w_{01}^2+2w_{02}w_{12}w_{01}) + w_{12}^2\sigma_1^2 + \sigma^2_2} (X_2 - \bb{E}(X_2))
\end{align*}
\begin{align*}
    \mathrm{Var}(Y\vert X_2) &=  \mathrm{Var}(Y) - \frac{ \mathrm{Cov}(Y, X_2)^2}{ \mathrm{Var}(X_2)} \\
    &= w_{01}^2\sigma^2_0 + \sigma^2_1 - \frac{(\sigma_0^2(w_{01}w_{02}+w_{01}^2w_{12})+w_{12}\sigma^2_1)^2}{\sigma_0^2(w_{02}^2+w_{12}^2w_{01}^2+2w_{02}w_{12}w_{01}) + w_{12}^2\sigma_1^2 + \sigma^2_2}
\end{align*}

If we additionally assume $\mu_i = 0, w_{ij} = 1 \; \forall i,j$ and $\sigma^2_1=\sigma^2_2=1$, the above expressions become
$$
    \bb{E}(Y\vert X_2) = \frac{1}{2} X_2 \qquad \text{and} \qquad \mathrm{Var}(Y\vert X_2) = \sigma^2_0 + 1 - \frac{(2\sigma^2_0+1)^2}{4\sigma^2_0+2} = \frac{1}{2}.
$$
Consider now an intervention on $X_0$. We have that $S_0=\{0\}$ is intervention stable and a set of plausible causal predictors. On the other hand, $S_2 = \{2\}$ does not d-separate $Y$ from the intervention on $X_0$, and is not intervention stable; however, for interventions that affect only the variance of $X_0$ (i.e.\ $\sigma^2_0$), $S_2$ is a set of plausible causal predictors. Under this setting, we have that $\bb{S}_\cc{E} \subset \bb{C}_\cc{E}$.
\end{example}

\autoref{example_stable_vs_plausible} shows that $\bb{S}_\cc{E} \neq \bb{C}_\cc{E}$. However, one might ask how often this happens in practice. In the example, this only happens when we set the weights, means and variances to very particular values. When these parameters are sampled from a continuous distribution, we conjecture that the set of parameters for which $\bb{S}_\cc{E} \neq \bb{C}_\cc{E}$ has probability zero. We call the assumption that $\bb{S}_\cc{E} = \bb{C}_\cc{E}$ \emph{stability-faithfulness}.

\subsection{Informative interventions} If we make the assumption that $\bb{C}_\cc{E} = \bb{S}_\cc{E}$, by \autoref{lem:parents_stable} we know that, in the absence of estimation errors,  a direct intervention on a parent is sufficient for it to appear in the ICP estimate. However, it is not a necessary condition, as is shown in the following example.

\begin{center}
\begin{minipage}{.50\textwidth}
\begin{example}
\label{example_dir_intervention}
Let $\cc{E}$ be a collection of two environments: one without interventions and one with a direct intervention on $X_2$, as shown in the graph. The intervention stable sets are
\begin{align*}
\bb{S}_\cc{E} =&\{0, 1, 2\},\\
    &\{0, 1, 3\},\\
    &\{0, 1, 2, 3\}.
\end{align*}
Therefore,
$$S(\cc{E}) = \bigcap_{S: S\in \bb{S}_\cc{E}} S = \{0, 1\},$$
\end{example}
\end{minipage}
\begin{minipage}{.49\textwidth}
\begin{figure}[H]
    \centering
    \begin{tikzpicture}
        \SetVertexStyle[FillColor=white, MinSize=0.8\DefaultUnit, TextFont=\small]
        \SetEdgeStyle[Arrow=-stealth, TextFont=\small, LineWidth=1]
        \Vertex[y=.9, x=-1.8,label=$X_0$]{0}
        \Vertex[y=.9, x=1.8, label=$X_1$]{1}
        \Vertex[y=0, label=$X_2$]{2}
        \Vertex[y=-1.35,label=$X_3$]{3}
        \Vertex[y=-2.7, label=$Y$]{4}
        \Vertex[y=1.5, x=0, label=$I$, shape=diamond, style=white, fontcolor=red]{I}
        
        \Edge[Direct, color=red](I)(2)
        \Edge[Direct](0)(2)
        \Edge[Direct](0)(4)
        \Edge[Direct](1)(2)
        \Edge[Direct](1)(4)
        \Edge[Direct](2)(3)
        \Edge[Direct](3)(4)
    \end{tikzpicture}
\end{figure}
\end{minipage}
\end{center}
which shows that parents can appear in the intersection of intervention stable sets without being directly intervened on.
In this case, a direct intervention on $X_2$ is very informative, as it reveals two parents simultaneously. To the best of our knowledge it is not clear when situations like the above arise, or how they can be detected from the accepted sets. Therefore, as a first approach we consider direct interventions on the parents as ``maximally informative'', and the goal of the proposed policies is to pick such interventions.

\section{Detailed description of ICP}\label{apx:alg}

Here we present a slightly adapted version of Invariant Causal Prediction \cite{peters2016causal}. In contrast to the original formulation, \autoref{alg:icp_generic} takes \texttt{candidate sets} as an additional, optional argument. If \texttt{candidate sets} is not provided, \autoref{alg:icp_generic} corresponds to the original ICP formulation where the null hypothesis $H_{0,S}$ needs to be tested for all subsets of the predictors. As detailed in \autoref{cor_speedup}, \aicp (algorithm \autoref{alg_active_icp}) does not require testing all subsets in each iteration. Hence, when ICP is called as a subroutine in \aicp only the \texttt{accepted sets} from the previous iteration are provided as \texttt{candidate sets} to ICP. 

In general terms, the null hypothesis $H_{0,S}$ states that the distribution of the response $Y$ conditional on the predictors $X_S$ is invariant across the different environments. Depending on which ICP version is employed, the specific formulation of null hypothesis is adapted to the respective problem setting. In the linear case, one can test for the equality of the regression coefficients and the noise variances across environments but other options are also possible (for details, please see \cite{peters2016causal}). When using nonlinear ICP \cite{heinze2018invariant}, the environment is considered as an additional variable $E$ in the system and the null hypothesis then corresponds to $Y \indep E \vert X_S$ which is tested using a non-parametric conditional independence test. To formulate the algorithm below generically, we leave open what formulation and test is chosen for $H_{0,S}$. 

\begin{algorithm}[H]
\caption{ICP} \label{alg:icp_generic}
\SetAlgoLined
\SetKwInOut{Input}{Input}
\SetKwInOut{Output}{Output}
\Output{\texttt{accepted sets} sets for which the null hypothesis cannot be rejected,\newline 
$\hat{S}(\cc{E})$ estimate of the parents of the response}
\Input{i.i.d.\ samples of ($X, Y$) from different environments $\cc{E}$,\newline 
\texttt{candidate sets} sets for which to test the null hypothesis,\newline
$\alpha$ significance level}
\If{\emph{\texttt{candidate sets} is \texttt{null}}}{
\texttt{candidate sets} $\leftarrow S \subseteq \{1, \ldots, p \}$}
\For{\emph{\textbf{each}} $S$ in \emph{\texttt{candidate sets}}}{
Test whether $H_{0,S}$ holds at level $\alpha$.
}
$\hat{S}(\cc{E}) := \bigcap_{S:H_{0,S} \text{ not rejected}} S $\;
\texttt{accepted sets} $\leftarrow \{S \mid H_{0,S} \text{ not rejected}\}$\;
\Return \texttt{accepted sets}, $\hat{S}(\cc{E})$
\end{algorithm} 

\section{Analysis of computational complexity}
\label{apx:complexity}
The runtime of \aicp depends on (i) the runtime of ICP, and (ii) the runtime of the chosen  intervention selection policy. The runtime of ICP depends on the complexity $c(N,e,k)$ of testing the invariance from Eq.~\eqref{icp_hypothesis} for a set of predictors $S$ of size $k$ over a total of $N$ observations from $e$ environments. 
Thus, for each iteration of \aicp, the cost of running ICP on all sets of predictors is
$$\sum_{S\subseteq \{1,...,p\}} c(N,e,|S|) \leq 2^pc(N,e,p).$$
Furthermore, let $s(N,e,p)$ denote the complexity of the chosen intervention selection policy. In this notation, the complexity of running \aicp for a total of $T$ iterations is
\begin{align}
    \label{eq:complexity}
    O(T2^pc(N,e,p)s(N,e,p)).
\end{align}

In the experiments of section \ref{sec:experiments}, we test invariance by performing a least-squares regression of the response on the predictors, and then running a two-sample t-test and an F-test \cite[section 3.1.2]{peters2016causal} over the residuals. Under this approach, the complexity of testing a single set of predictors of size $k$ is the cost of performing a least-squares regression and computing the residuals ($O(k^2N)$) and the cost of performing the t-test and F-test over each split of the $e$ environments ($O(eN)$). Thus, $c(N,e,k) = O(N(k^2 + e))$.

The cost of the empty-set strategy corresponds to that of testing the empty set over the initial and current environments, i.e. $c(N,2,0)$. For the ratio strategy, one must compute the stability ratio (c.f. \autoref{def:stab_ratio}), which in the worst case (all sets are accepted) incurs a cost of $O(p2^p)$. The Markov strategy carries the cost of performing a Lasso regression in the first iteration, which is dominated by the other terms in \autoref{eq:complexity}.

By \autoref{cor_speedup}, at each iteration it suffices for ICP to consider only the sets accepted in the previous iteration. On average, this provides a substantial speed up as not all $2^p$ subsets of predictors need to be re-tested in each iteration. However, the complexity is still exponential in the number of variables $p$, which limits the applicability of \aicp to ``large $p$'' settings. Nonetheless, \aicp can still be useful in settings where the time needed to carry out an experiment far outweighs the computation time to select the next experiment, which is common in empirical sciences like biology.

\section{Additional experimental results}
\label{apx:figures}
Here, we present additional experimental results. In \autoref{apx:figures:fwer_avg}, we show the average number of interventions until exact recovery (\autoref{fig_20_iter_int_numbers_finite}) for the finite-sample experiments presented in \autoref{sec:experiments}. In \autoref{apx:figures:50_iter}, we provide additional results for the total 50 iterations over which the policies are run: the family-wise error rate is shown in \autoref{fig_fwer}, Figures~\ref{fig_50_iter_jaccard_finite} and~\ref{fig_50_iter_int_numbers_finite} show the Jaccard similarity and the average number of iterations until exact recovery, respectively. The error analysis of the Markov blanket estimation procedure is displayed in \autoref{apx:figures:mb_estimation}, \autoref{fig_mb_estimation}. In \autoref{apx:figures:abcd}, we present the results from running ABCD and \aicp with different sizes of the initial observational sample. Finally, \autoref{apx:figures:level_strength} contains additional results comparing the interplay between the \aicp significance level and the performance for different intervention strengths.

\subsection{Average number of interventions for exact recovery}
\label{apx:figures:fwer_avg}

\begin{figure}[H]
\centering
\includegraphics[width=1.0\textwidth]{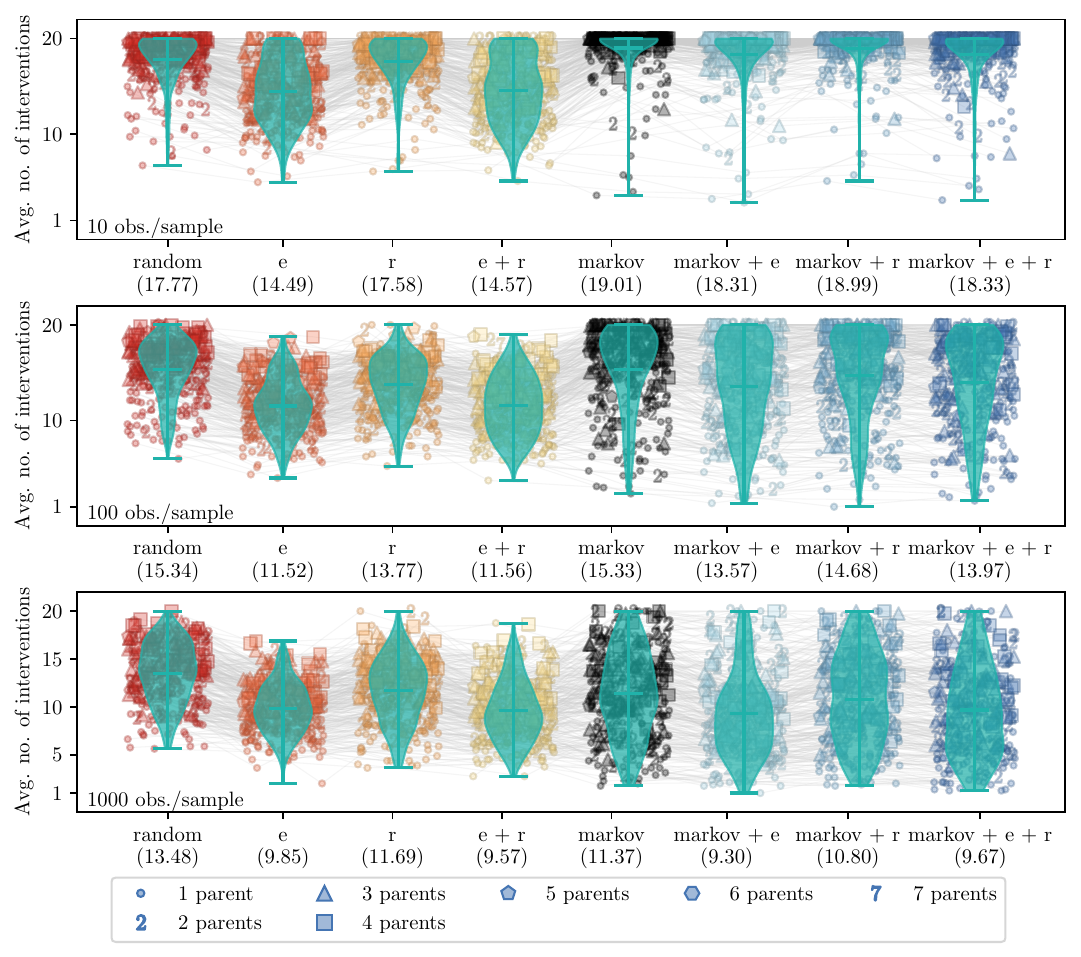}
\caption{(finite regime) Average number of interventions until the causal parents are recovered exactly for the first 20 iterations of \aicp, for each one of the 300 SCMs. If $\hat{S}\neq S^*$ at $t=20$, we set the statistic to 20. For each SCM, we average the performance over the 8 different random seeds considered. The average performance of each SCM is represented by a dot and connected across policies by a grey line. The total average of interventions employed by each policy is given below its label. The ``e'' policy performs well across all sample sizes, and is the best performer except at 1000 obs./sample where it falls behind the ``e + r'' and ``Markov + e`` policies.}
\label{fig_20_iter_int_numbers_finite}
\end{figure}


\subsection{Results for 50 interventions}
\label{apx:figures:50_iter}

We run the policies for a total of 50 interventions, to evaluate their performance in a setting where more experimental rounds are possible.

\subsubsection{Family-wise error rate}\label{apx:50_iter_fwer}
In \autoref{fig_fwer} we plot the family-wise error rate (FWER) $\hat{P}(\hat{S} \not\subseteq S^*)$. Recall that to achieve  FWER control across all iterations, we have to apply a correction to the level at which ICP is run in each iteration of \aicp (also see \autoref{apx:alg}, algorithm \autoref{alg_active_icp}). Due to the strong dependence between the tests, we use a Bonferroni correction by running ICP at iteration $t$ at the level $\alpha/T$ where $\alpha$ is the overall significance level and $T$ is the total number of iterations. \autoref{fig_fwer} confirms that the FWER is indeed kept below the $0.01$ significance level at which \aicp is run, maintaining the coverage guarantees provided by Invariant Causal Prediction (ICP). The FWER lies well below the nominal level of $0.01$ due to the construction of the estimate $\hat{S}$. The error control rests on the fact that the true set of causal parents is rejected with probability smaller than $\alpha/T$ in each round of \aicp. However, even if a mistake is made and the true set is rejected, accepting other sets and computing their intersection to obtain $\hat{S}$ may still result in $\hat{S} \subseteq S^*$.

\begin{figure}[H]
\centering
\includegraphics[width=1.0\textwidth]{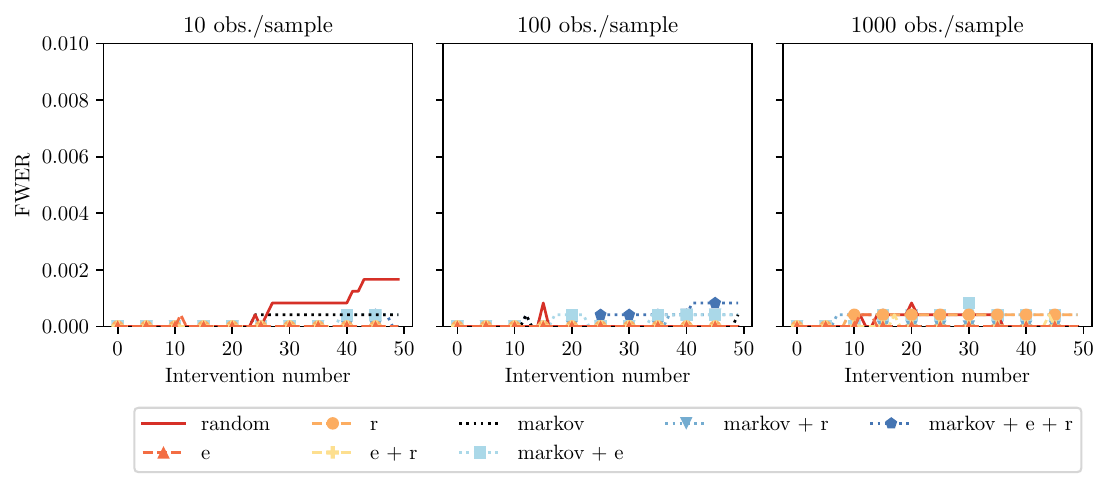}
\caption{(finite regime) Family-wise error rate (FWER) for the finite-sample experiments. The FWER $\hat{P}(\hat{S} \not\subseteq S^*)$, i.e. the probability of wrongly marking as direct causes variables which are not, is kept below the $0.01$ significance level at which \aicp is run, maintaining the coverage guarantees provided by Invariant Causal Prediction (ICP) \cite{peters2016causal}. }
\label{fig_fwer}
\end{figure}

\subsubsection{Jaccard similarity and average number of interventions for exact recovery}\label{apx:50_iter_jacc}
The results in \autoref{fig_50_iter_int_numbers_finite} and \autoref{fig_50_iter_jaccard_finite} illustrate the fact that if there are no constraints on the number of interventions, the \emph{random} policy is among the most robust options, as its choice of intervention targets is unaffected by estimation errors. However, it needs a large number of iterations to achieve competitive performance and only achieves an average Jaccard similarity close to one when $t$ approaches 50. 

Overall, the \emph{empty-set} strategy is the best performer across all sample sizes for a large range of intervention numbers. For the \emph{Markov} policies, the issues arising from obtaining an estimate of the Markov blanket are more apparent in this setting: while the policies quickly identify parents contained in the estimate, they become stuck performing non-informative interventions and fail to identify the remaining parents for some SCMs. This can be seen in \autoref{fig_50_iter_int_numbers_finite} which shows the average number of interventions needed to achieve exact recovery (averaged over different random seeds). 

While at 1000 obs./sample, combining the \emph{ratio} strategy with the \emph{empty-set} strategy grants an advantage in performance over using the \emph{empty-set} strategy alone for the early iterations, this advantage is lost later on as the combination performs worse for some particular graphs, which decreases the average performance. Combining the \emph{ratio} with the \emph{empty-set} strategy can in some cases be less effective than the \emph{empty-set} strategy alone for the following reasons. First, the interventions chosen here are quite strong such that the \emph{empty-set} strategy is \emph{not} affected by the issue of statistical power that the empty set may be wrongly accepted after an intervention on a parent, which would then be discarded from future interventions. Second, in the finite regime, it is not necessarily sufficient to intervene on a parent once for it to appear in $\hat{S}$ due to a lack of power. In other words, after an intervention on a parent not all unstable sets are necessarily rejected. In contrast, intervening on children of the response can sometimes lead to a larger number of unstable sets being rejected and hence an estimate $\hat{S}$ with larger Jaccard similarity. Intervening on children of the response tends to occur more often when using the \emph{empty-set} strategy alone. Lastly, the \emph{ratio} strategy is subject to the following testing errors: falsely rejecting stable sets and wrongly accepting unstable sets can bias the estimate of the stability ratio of some parents and thereby keep \aicp from intervening on such parents. Since a rejected set is not re-tested at future iterations (by \autoref{cor_speedup}), falsely rejecting a stable set at some iteration $t$ will also bias the estimate of the stability ratio for future iterations (as long as the set remains stable). While this discussion highlights the failure cases of the \emph{ratio} strategy, the analysis in \autoref{apx:figures:level_strength} shows that for smaller intervention strengths the \emph{empty-set} strategy is not always the best-performing policy, presumably due to the power issue described above.  

\begin{figure}[H]
\centering
\includegraphics[width=1.0\textwidth]{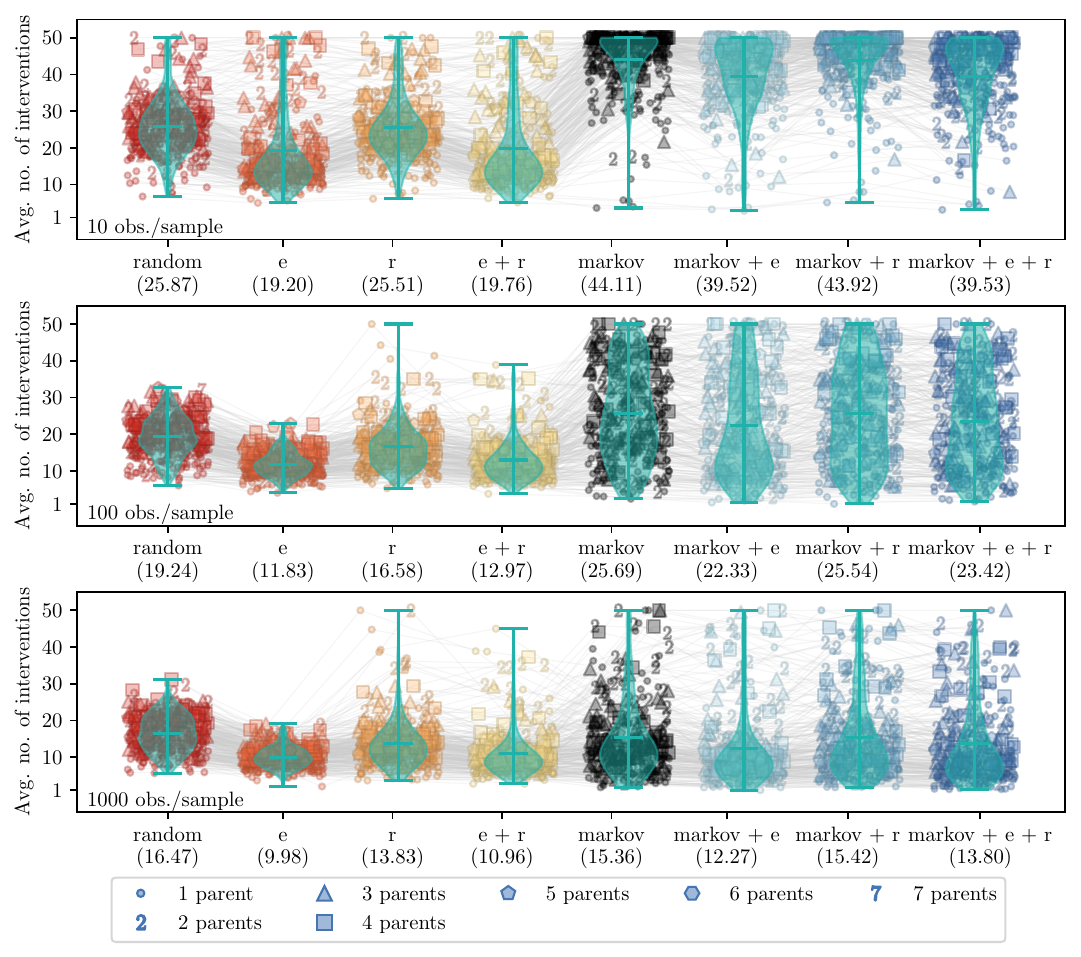}
\caption{(finite regime)  Average number of interventions until the causal parents are recovered exactly for $T=50$, for each one of the 300 SCMs. If $\hat{S}\neq S^*$ at $t=50$, we set the statistic to 50. For each SCM, we average the performance over the 8 different random seeds considered. Each SCM is represented by a dot and connected across policies by a grey line. The total average of interventions employed by each policy is given below its label.}
\label{fig_50_iter_int_numbers_finite}
\end{figure}

\begin{figure}[H]
\centering
\includegraphics[width=1.0\textwidth]{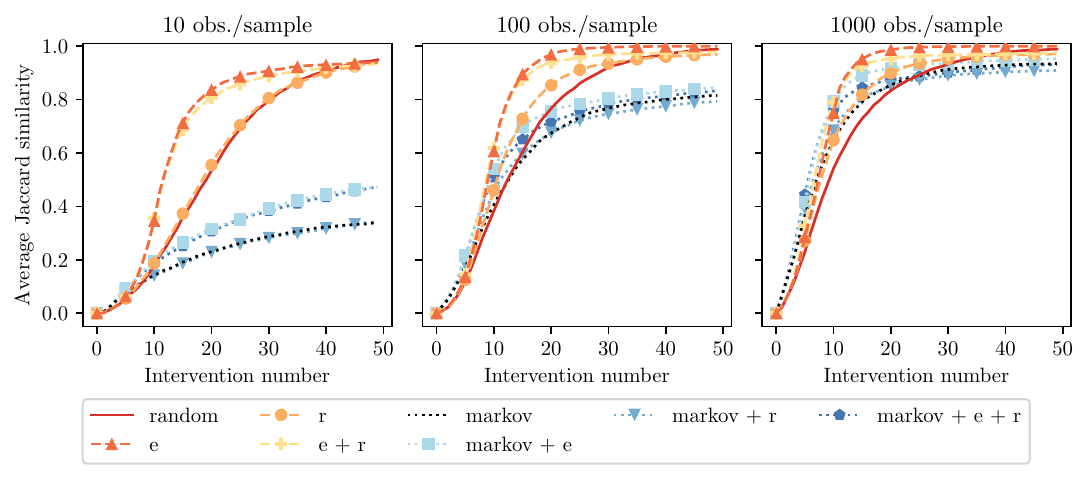}
\caption{(finite regime) Average Jaccard similarity for 300 SCMs of size 12 as a function of the number of interventions for 10, 100 and 1000 observations per sample. Here, $\alpha = 0.01$ and the policies' performance is shown for all $50$ iterations.}
\label{fig_50_iter_jaccard_finite}
\end{figure}

\subsection{Effect of the observational sample size on the performance of ABCD and \aicp}
\label{apx:figures:abcd}

\begin{figure}[H]
\centering
\includegraphics[width=1.0\textwidth]{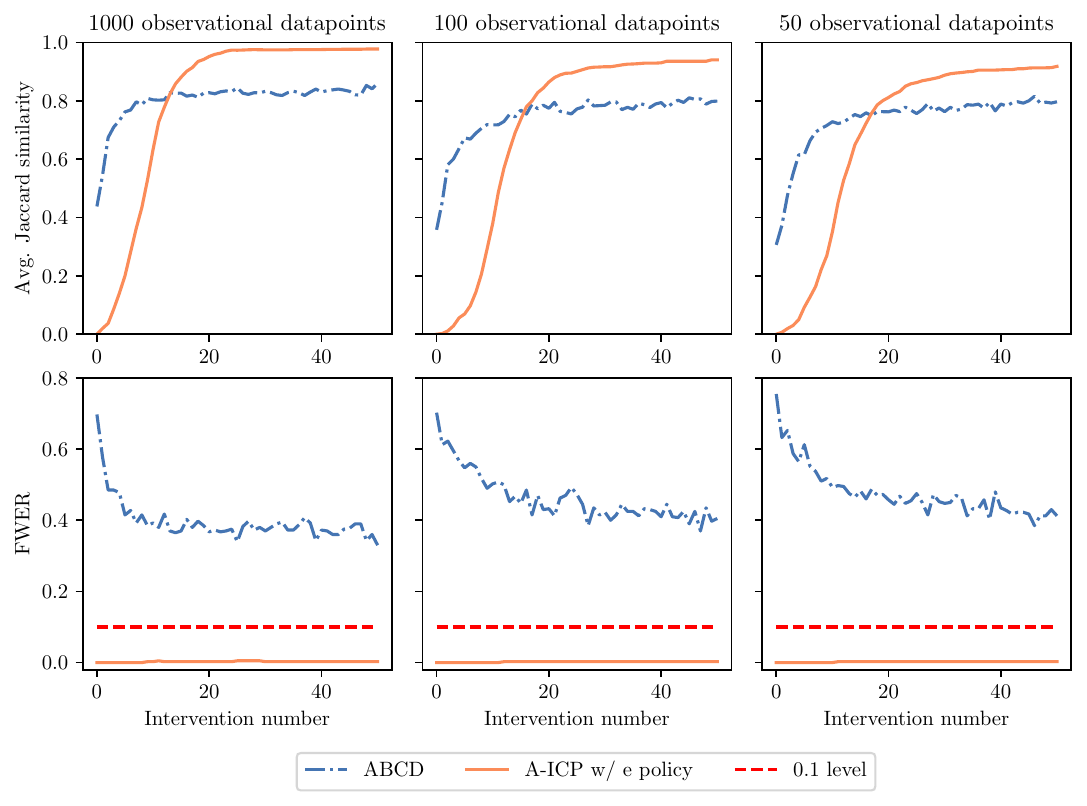}
\caption{Average Jaccard similarity (top) and family-wise error rate (FWER) (bottom) of the ABCD and \aicp estimates, as a function of the number of interventions (50 in total). Both the  Jaccard similarity and family-wise error rate of ABCD are affected by the size of the initial observational sample. While the performance of \aicp improves slightly with a larger observational sample, having a large observational sample is not a requirement of \aicp. In all cases, \aicp maintains the FWER under the desired level ($\alpha = 0.1$) while its power increases steadily with the number of interventions. Details about the experimental setup can be found in \autoref{apx:experiments}.}
\label{fig_abcd_vs_aicp_all}
\end{figure}

For the results summarized in \autoref{fig_abcd_vs_aicp_all}, we vary the size of the initial observational sample while keeping the number of observations per interventional sample fixed at 10. For the observational sample we consider 50, 100 and 1000 observations. ABCD requires a large observational sample to obtain a sufficiently good estimate of the posterior over graphs. This leads to relatively large Jaccard similarities for the first few iterations. In contrast, \aicp remains conservative at the beginning, often returning the empty set as an estimate as a large number of predictor sets are stable for small $t$. While \aicp controls the nominal FWER of $\alpha = 0.1$ over all iterations, its power increases steadily with the number of interventions, reaching an average Jaccard similarity close to one for large $t$. In contrast, ABCD does not control the false positives: while the average Jaccard similarity increases with the number of iterations, it does not approach one since the estimate still contains false positives even for large $t$. The comparison for different observational sample sizes shows that both the average Jaccard similarity and the FWER improve for ABCD the larger the initial observational sample is.

\subsection{Error analysis of the Markov blanket estimation procedure}
\label{apx:figures:mb_estimation}

\begin{figure}[H]
\centering
\includegraphics[width=0.82\textwidth]{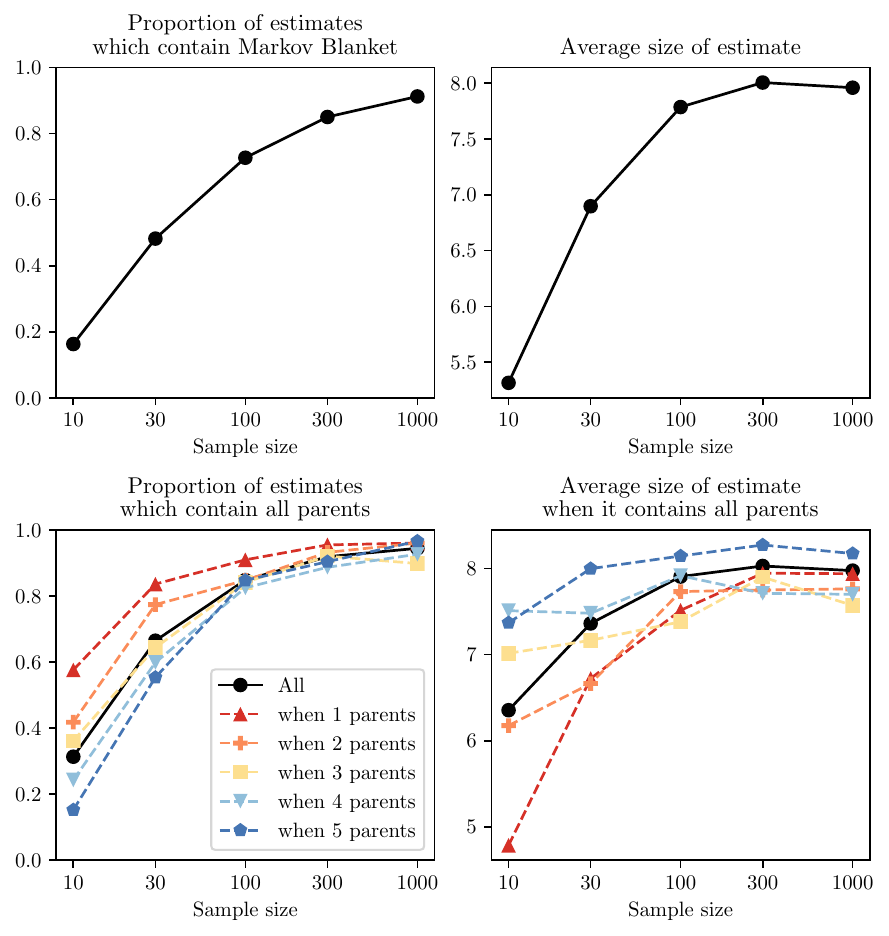}
\caption{Error analysis of the Markov blanket estimation procedure for the finite regime, for 1000 SCMs of size 12. The estimate is produced by taking the variables with non-zero coefficients resulting from a Lasso regression over all predictors in the sample from the observational environment. The regularization parameter is picked individually for each SCM by ten-fold cross-validation. By size we refer to the number of variables included in the estimate. As expected, the quality of the estimate improves with the sample size. However, even at the largest sample size, for some SCMs not all the parents are contained in the estimate. In these cases the policies relying on the \emph{Markov} strategy become stuck performing non-informative interventions, and fail to recover all parents after the limit $T$ of iterations is reached.}
\label{fig_mb_estimation}
\end{figure}

In figure \ref{fig_mb_estimation} we provide further analyses to understand the behavior of the policies using the \emph{Markov} strategy. For 1000 SCMs of size 12, the Markov blanket is estimated with the Lasso using the observational sample. The regularization parameter is chosen by ten-fold cross-validation. In \autoref{fig_mb_estimation} we plot (i) the proportion of estimates which contain the Markov blanket (top left); (ii) the average size of the estimate where size refers to the number of variables included in the estimate (top right); (iii) proportion of estimates which contain all parents (bottom left); and (iv) the average size of the estimate when it contains all parents (bottom right). At smaller sample sizes, often not all parents are included in the estimate (bottom left). Hence, policies using the \emph{Markov} strategy do not intervene on them which often results in a failure to identify them. While this issue is attenuated for larger sample sizes, it does not disappear entirely, even for a sample size of 1000. This explains why the policies using the \emph{Markov} strategy have a lower average Jaccard similarity for large $t$, as can be seen in \autoref{fig_50_iter_jaccard_finite}.

\subsection{\aicp significance level and intervention strength}
\label{apx:figures:level_strength}

\begin{figure}[H]
\centering
\includegraphics[width=1.0\textwidth]{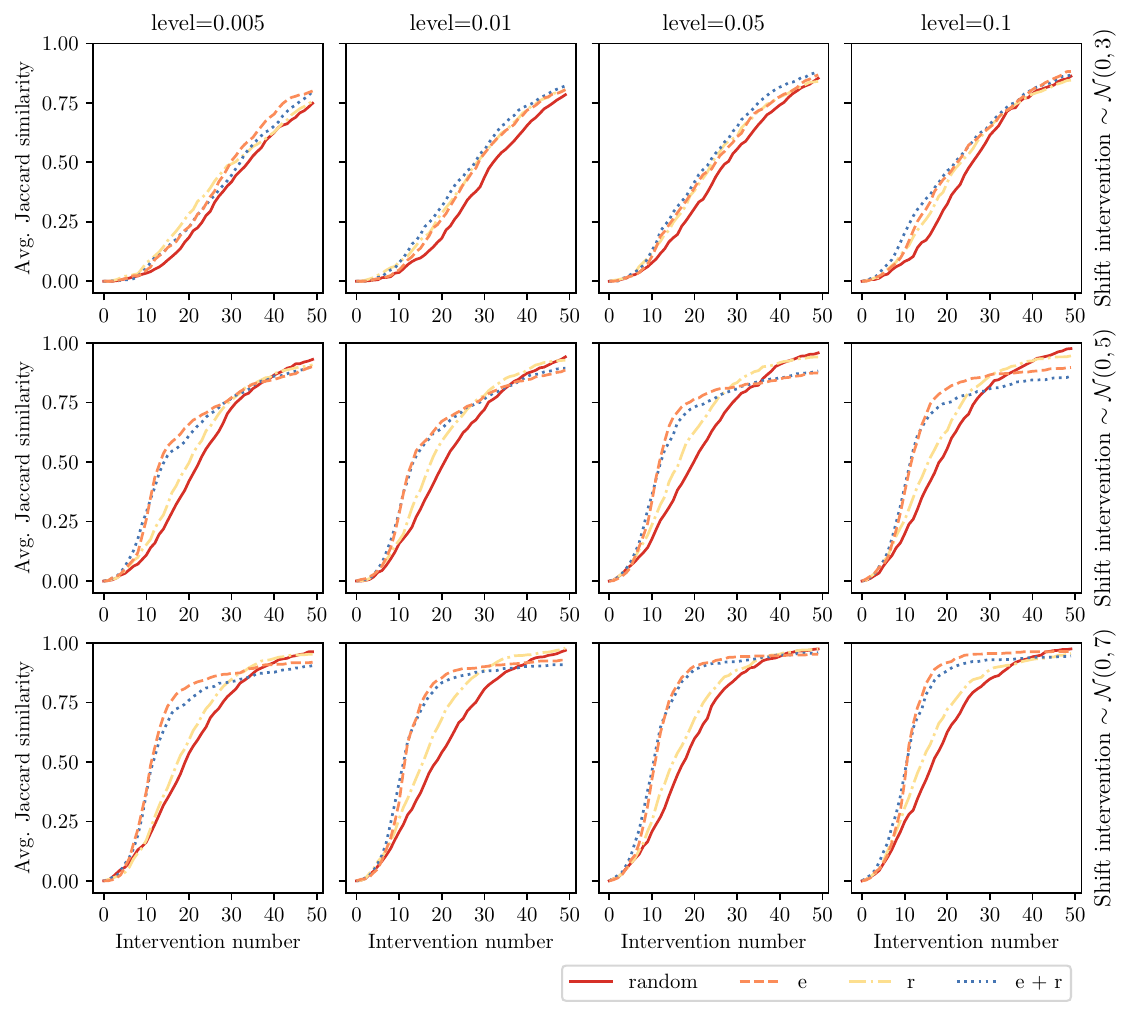}
\caption{(finite regime) Performance of the policies for an observational sample of size 100 and 10 observations per interventional sample, for different intervention strengths (rows), and significance levels of \aicp (columns). Overall, the gains in performance over the \emph{random} policy increase with the intervention strength. The performance of the \emph{empty-set} policy increases with the level, as power also increases and the empty set is rejected more often. The \emph{ratio} policy is largely unaffected by the change in level, and often yields additional improvements when used in combination with the \emph{empty-set} strategy in the initial iterations. 
} 
\label{fig_strength}
\end{figure}

To correct for multiple testing of the accepted sets, we apply a Bonferroni correction to the significance level of the statistical tests performed in each round of \aicp (see algorithm \autoref{alg_active_icp}). To assess the sensitivity of the results with respect to the overall significance level of \aicp $\alpha$ and the intervention strength, we run \aicp at $100$ observational data points and $10$ observations per interventional sample for $\alpha \in \{0.005, 0.01, 0.05, 0.1 \}$, and shift interventions with variance $1$ and means $3,5$ and $7$. Details about the experimental settings can be found in \autoref{apx:experiments}.

\autoref{fig_strength} shows that the \emph{ratio} and the \emph{empty-set} strategy yield larger improvements over the \emph{random} policy for larger intervention strengths. This is to be expected as statistical power increases with the intervention strength and both the \emph{ratio} and the \emph{empty-set} strategy rely on statistical testing to choose the intervention target. While the results reported in \autoref{sec:experiments} are based on experiments with strong interventions (shift interventions with variance $1$ and mean $10$), the relative performance between the \emph{ratio} and the \emph{empty-set} strategy changes when considering weaker interventions. For instance, for interventions with mean $5$ (second row), the \emph{empty-set} strategy does not reach an average Jaccard similarity of one for $t=50$. For large $t$, the \emph{ratio} as well as the \emph{random} strategy perform better.

\section{Experimental settings}
\label{apx:experiments}

\begin{table}[h]
\caption{Overview of the experimental parameters considered. Below, $n_e$ denotes the number of interventions per interventional sample. For the ABCD experiments, we additionally vary the size of initial observational sample. Unless indicated otherwise, all interventions are shift interventions with different means $\mu$ and variance $\sigma^2=1$.
\label{tab:exp_settings}}
{\small 
\begin{center}
{\renewcommand{\arraystretch}{1.4}
\begin{tabular}{c| ccccccc}
& \# SCMs & \# seeds & $n_e$ & $p$ & $T$ & $\alpha$ & Interventions \\
\hline
Population & 1000& 8& --- & 15& 15& ---  & $\mu=10$\\
Finite regime & 300 & 8 &  $\{10, 100, 1000\}$ & 12 & $50$ &  $.01$ & $\mu=10$ \\
ABCD & 100& 4& $10$ & 12& 50 & $.1$ & do, $\mu=9$\\
\autoref{fig_strength} & 100 & 4 & $10$ & 12 & $50$ &  $\{.005, .01, .05, .1\}$ & $\mu\in\{3,5,7\}$ \\
\end{tabular}}
\end{center}
}
\end{table}

The code to reproduce the experimental results is provided in the repositories \url{https://github.com/juangamella/aicp} and \url{https://github.com/juangamella/abcd}. Additionally, to generate synthetic interventional data we make the python package \texttt{sempler} (\url{https://github.com/juangamella/sempler}) available. 

\paragraph{Population setting} For the experiments, 1000 linear structural causal models of size 15 are randomly generated by sampling from Erd\H{o}s-R\'enyi graphs with an average degree of 3. The weights are sampled uniformly at random from $[0.5,1]$, and the intercepts and noise variances from $[0,1]$. In the population setting no further assumptions are made on the noise distributions, besides having finite mean and variance to perform the OLS regression. To perform the regression in the population setting, we maintain a symbolic representation of distributions that contains their first and second moments, and allows conditioning and marginalization. Further experiments with SCMs of different size and parameters yielded very similar results to the ones presented in the main text and are not shown separately. For every SCM, each policy is run 8 times with different random seeds, to account for the stochastic component of the policies.

\paragraph{Finite sample setting} For the experiments, 300 linear structural causal models of size 12 are randomly generated, again by sampling from Erd\H{o}s-R\'enyi graphs with an average degree of 3. The weights, variances and intercepts are sampled as in the population setting. Interventions are shift-interventions with mean 10 and variance 1. Like in the population setting, the policies are run 8 times with different random seeds, for 50 iterations. To simplify the implementation, we assume that the underlying noise distributions are Gaussian, and set ICP to use a two-sample t-test and F-test to check the invariance of the conditional distribution of the response. It is important to note that this is not a necessary requirement: the results derived in section 2 (e.g.\ \autoref{cor:ratio_parents}) apply to arbitrary SCMs with arbitrary noise distributions, and ICP can use other statistical tests, including non-parametric ones. However, we expect that the effect of the sample size on the results will be different under different noise distributions and tests. \autoref{fig_20_iter_jaccard_finite} corresponds to the results of running \aicp at a significance level of $0.01$.

\paragraph{Comparison with ABCD} We randomly generate 100 linear structural causal models of size 12, by sampling from Erd\H{o}s-R\'enyi graphs with an average degree of 3. The weights, variances and intercepts are sampled as in the population setting. ABCD requires a Gaussian SCM, so the underlying noise distributions are Gaussian and ICP is set to use a two-sample t-test and F-test to check the invariance of the conditional distribution of the response. At each iteration, each method receives 10 observations from the newly performed intervention. Interventions are do-interventions, as this is the only type of intervention that the ABCD implementation considers. Experiments are carried out for different sizes of the initial observational sample (see \autoref{fig_abcd_vs_aicp_all}), running each method a total of 4 times to account for stochasticity. The output of ABCD are posterior probabilities over parent sets; the average Jaccard similarity and FWER are computed by taking the argmax of the posterior. ABCD is set to use $100$ bootstrap samples and \aicp is run at a significance level of $0.1$.

\paragraph{Intervention strength vs. level (\autoref{fig_strength})} The experiments are run on 100 randomly generated linear structural causal models of size 12, sampled from Erd\H{o}s-R\'enyi graphs with an average degree of 3. The remaining parameters are sampled as in the population setting. We then compare the performance of the \emph{random}, \emph{empty-set} and \emph{ratio} strategies at different significance levels ($0.005, 0.01, 0.05$ and $0.1$) and intervention strengths, i.e. we use shift interventions with variance $1$ and means $3,5$ and $7$. We collect $100$ observations from the initial observational environment and $10$ observations from each interventional environment. Again ICP employs a t-test and F-test to check the invariance of the conditional distribution of the response. 

\section{Proofs}
\label{apx:proofs}

To simplify notation, let $\text{PA}(i)$ be the parents of $X_i$ and let $\text{PA}(S) = \{j \in \{1,...,p\} \mid \exists i \in S: j \in \text{PA}(i)\}$ denote the parents of variables in a set $S$. Similarly, let $\text{CH}(i)$ be the children of $X_i$ and let $\text{CH}(S) = \{j \in \{1,...,p\} \mid \exists i \in S: j \in \text{CH}(i)\}$ denote the children of variables in a set $S$. Let $\text{DE}(S) = \{j \in \{1,...,p\} \mid \exists i \in S: j \in \text{DE}(i)\}$ denote the descendants of variables in a set $S$. Note that the descendants of a variable include the variable itself, i.e.\ $i \in \text{DE}(i)$.

\lemmaparents*
\begin{proof}
Assume $S \subseteq \{1,...,p\}$ is an intervention stable set such that $j \notin S$, and let $I$ denote the direct intervention on $j$. Then, there is a path $I \rightarrow j \rightarrow Y$ that is unblocked by $S$, which contradicts $S$ being intervention stable.
\end{proof}

\lemmachildren*
\begin{proof}
Let $I$ denote the direct intervention on $i$, and let $S \subseteq \{1,...,p\} : S \cap \text{DE}(i) \neq \emptyset$. Then, the path $Y \rightarrow i \leftarrow I$ is not blocked by $S$.
\end{proof}

\lemmaempty*
\begin{proof}
($\implies$) Assume the empty set is stable under environments $\cc{E}$ which contain an intervention $I$ on $j \in \text{AN}(Y)$. Then there exists a path $Y \leftarrow ... \leftarrow j \leftarrow I$ which is not blocked by the empty set, arriving at a contradiction. ($\impliedby$) For every intervention $I$ on a variable $i$, every path from $Y$ to $I$ either
\begin{enumerate}[label=(\roman*)]
    \item contains a collider, and is thus blocked by $\emptyset$, or
    \item does not contain a collider and is active under $\emptyset$.
\end{enumerate}
Since $I$ is a source node, paths of type (ii) can only be of the form $Y \leftarrow ... \leftarrow i \leftarrow I$,  which is not possible as $i$ would then be an ancestor of $Y$.
\end{proof}

\propratio*
\begin{proof}
We will prove the equivalent statement $j \in \text{AN}(Y) \implies r_{\cc{E}}(j) \geq 1/2$. For any $i \in \{1,...,p\}$ we have that
$$r_\cc{E}(i) = \frac{|\{S \in \bb{S}_\cc{E} : i \in S\}|}{|\{S \in \bb{S}_\cc{E} : i \in S\}| + |\{S \in \bb{S}_\cc{E} : i \notin S\}|},$$
and therefore
\begin{align}\label{eq_proof_ratio_parents}
    r_\cc{E}(i) \geq 1/2 \iff |\{S \in \bb{S}_\cc{E} : i \in S\}| \geq |\{S \in \bb{S}_\cc{E} : i \notin S\}|.
\end{align}
We will show that for any $j\in \text{AN}(Y)$, and any intervention stable set $S$ such that $j \notin S$, the set $S \cup \{j\}$ is also intervention stable, satisfying the right hand side of \autoref{eq_proof_ratio_parents}. To do this, we will use the fact that $S$ d-separates the response from all interventions, and show that the same is true for $S \cup \{j\}$, making it intervention stable.

Let $I$ denote an intervention on a variable $i$. For every path connecting $Y$ and the intervention, either

\begin{enumerate}[label=(\roman*)]
    \item $j$ appears in the path as a collider,
    \item $j$ appears in the path but not as a collider,
    \item $j$ does not appear in the path but is downstream of a collider, or
    \item $j$ does not appear in the path and is not downstream of a collider.
\end{enumerate}
If $S$ blocks paths of type (ii) and (iv), $S \cup \{j\}$ also does. Assume now there is a path of type (i) or (iii) which is blocked under $S$ but active under $S \cup \{j\}$. This implies that such path is blocked by a collider $c$ such that $j \in \text{DE}(c)$ and $S \cap \text{DE}(c) = \emptyset$; thus, there exists a path $Y \leftarrow ... \leftarrow j \leftarrow ... \leftarrow c \leftarrow ... \ i \leftarrow I$ which is active under $S$, i.e.\ $S \notin \bb{S}_\cc{E}$.

Therefore, for all $S \in \bb{S}_\cc{E}$ such that $j \notin S$, we have that $S \cup \{j\} \in \bb{S}_\cc{E}$, and
$$|\{S \in \bb{S}_\cc{E} : j \in S\}| \geq |\{S \in \bb{S}_\cc{E} : j \notin S\}| \implies r_\cc{E}(j) \geq 1/2.$$
\end{proof}

\propcausalpredictors*
\begin{proof}
The following is based on proof of proposition 3 in \cite{pfister2019stabilizing}.

Let $\cc{E}$ be a set of observed environments, and let $S\in\bb{S}_\cc{E}$ be an intervention stable set. From \cite{peters2016causal} we know that $S$ is a set of plausible causal predictors iff $Y^e|X^e_S$ remains invariant for all environments $e\in\cc{E}$. Starting from setting 1, introduce an auxiliary random variable $E$ taking values in $\cc{E}$ with equal probability (for simplicity). To model the environments we construct an extended SCM $\cc{S}^\cc{E}_{\text{full}}$, where the variable $E$ appears on the assignments of the intervention variables $I$, and the assignments of the remaining variables remain as in $\cc{S}^\cc{E}$. As such, in $\cc{G}(\cc{S}^\cc{E}_{\text{full}})$ $E$ is a source node with only edges into the variables in $I$. The SCM $\cc{S}^\cc{E}_{\text{full}}$ induces a distribution $P_\text{full}$ over $(E,I,X,Y)$, which under setting 1 has a density $p$ that factorizes with respect to a product measure. Furthermore, since $P_\text{full}$ satisfies the Markov properties \cite{pearl2009causal} and $S$ d-separates the response from all the intervention variables in $I$, it holds that $E \indep Y \mid X_S \sim P_{\text{full}}$. Therefore, for every environment $e \in \cc{E}$, we have that
\begin{align*}
p(Y^e = y \mid X^e_S = x) &= p(Y = y \mid X_S = x, E = e)
\\&= \frac{p(Y = y\mid X_S = x)p(E = e \mid X_S = x)}{p(E = e \mid X_S = x)}\\
&= p(Y = y \mid X_S = x),
\end{align*}
and $Y^e \mid X^e_S$ remains invariant for all environments $e\in\cc{E}$.
\end{proof}

\end{document}